\documentclass[11pt]{article}

\usepackage[letterpaper,margin=1in]{geometry}

\usepackage{lipsum}

\newcommand\blfootnotea[1]{%
  \begingroup
  \renewcommand\thefootnote{}\footnote{#1}%
  \endgroup
}

\usepackage[utf8]{inputenc} %
\usepackage[T1]{fontenc}    %

\usepackage{url}            %
\usepackage{booktabs}       %
\usepackage{nicefrac}       %
\usepackage{microtype}      %
\usepackage{enumitem}
\usepackage{dsfont}
\usepackage{multicol}
\usepackage{xcolor}
\usepackage{mdframed}
\usepackage{svg}
\usepackage{xspace}

\usepackage{amsthm,amsfonts,amsmath,amssymb,epsfig,color,float,graphicx,verbatim, enumitem,empheq,bm}

\usepackage{algorithmicx}
\usepackage[noend]{algpseudocode}
\usepackage{algorithm}

\usepackage{mathtools}
\usepackage{bbm}

\usepackage{thm-restate}

\definecolor{green}{rgb}{0.0, 0.5, 0.0}
\usepackage[colorlinks,citecolor=blue,linkcolor=magenta,bookmarks=true,hypertexnames=false]{hyperref}
\usepackage[nameinlink,capitalize]{cleveref}

\crefname{lemma}{lemma}{lemmata}
\crefname{claim}{claim}{claims}
\crefname{theorem}{theorem}{theorems}
\crefname{proposition}{proposition}{propositions}
\crefname{corollary}{corollary}{corollaries}
\crefname{claim}{claim}{claims}
\crefname{remark}{remark}{remarks}
\crefname{definition}{definition}{definitions}
\crefname{fact}{fact}{facts}
\crefname{question}{question}{questions}
\crefname{condition}{condition}{conditions}
\crefname{problem}{problem}{problems}
\crefname{algorithm}{algorithm}{algorithms}
\crefname{assumption}{assumption}{assumptions}
\crefname{notation}{notation}{notation}
\crefname{cond}{Condition}{Conditions}
\crefname{ineq}{Inequality}{Inequality}
\creflabelformat{ineq}{#2{\upshape(#1)}#3}
\crefname{sub}{Subsection}{Subsection}
\creflabelformat{Subsection}{#2{\upshape(#1)}#3}
\crefname{sdp}{SDP}{SDP}
\creflabelformat{sdp}{#2{\upshape(#1)}#3}
\crefname{lp}{LP}{LP}
\creflabelformat{lp}{#2{\upshape(#1)}#3}
\renewcommand\vec[1]{\mathbf{#1}}
\newcommand{\x}{\vec x}

\def\d{\mathrm{d}}

\newtheorem{theorem}{Theorem}[section]
\newtheorem{lemma}[theorem]{Lemma}

\newtheorem{claim}[theorem]{Claim}

\newtheorem{definition}[theorem]{Definition}

\newtheorem{fact}[theorem]{Fact}

\theoremstyle{definition}

\newtheorem{problem}[theorem]{Problem}

\renewcommand{\epsilon}{\varepsilon}
\newcommand{\eps}{\epsilon}

\newcommand{\poly}{\mathrm{poly}}

\newcommand{\Ind}{\mathds{1}}
\newcommand{\1}{\Ind}

\def\D{\mathcal D}

\def\R{\mathbb R}

\def\N{\mathbb N}

\def\Z{\mathbb Z}

\newcommand{\cA}{\mathcal{A}}
\newcommand{\cB}{\mathcal{B}}
\newcommand{\cC}{\mathcal{C}}
\newcommand{\cD}{\mathcal{D}}

\newcommand{\cN}{\mathcal{N}}

\newcommand{\cS}{\mathcal{S}}

\newcommand{\cU}{\mathcal{U}}

\newcommand{\cX}{\mathcal{X}}

\newcommand{\bI}{\vec{I}}

\newcommand{\bU}{\vec{U}}
\newcommand{\bV}{\vec{V}}

\newcommand{\bx}{\mathbf{x}}

\DeclareMathOperator*{\E}{\mathbf{E}}
\DeclareMathOperator*{\Var}{\mathbf{Var}}

\newcommand{\eqdef}{\stackrel{{\mathrm {\footnotesize def}}}{=}}

\newcommand{\op}{\textnormal{op}}
\newcommand{\fr}{\textnormal{F}}

\newcommand{\normal}{\mathcal{N}}
\def\d{\mathrm{d}}

\let\vec\mathbf

\def\colorful{0}

\ifnum\colorful=1
\newcommand{\new}[1]{{\color{red} #1}}

\else
\newcommand{\new}[1]{{#1}}

\fi

\allowdisplaybreaks

\title{Statistical Query Lower Bounds for Learning Truncated Gaussians \blfootnotea{Authors are in alphabetical order.}}

\author{
Ilias Diakonikolas\thanks{Supported by NSF Medium Award CCF-2107079, NSF Award CCF-1652862 (CAREER), a Sloan Research
Fellowship, and a DARPA Learning with Less Labels (LwLL) grant.
}\\
University of Wisconsin-Madison\\
{\tt ilias@cs.wisc.edu}\\
\and
Daniel M. Kane\thanks{Supported by NSF Medium Award CCF-2107547, NSF Award CCF-1553288 (CAREER), and a Sloan Research
Fellowship.}\\
University of California, San Diego\\
{\tt dakane@cs.ucsd.edu}
\and
Thanasis Pittas\thanks{Supported by NSF Medium Award CCF-2107079.}\\
University of Wisconsin-Madison\\
{\tt pittas@wisc.edu}\\
\and
Nikos Zarifis\thanks{Supported in part by DARPA Learning with Less Labels (LwLL) grant and NSF Award DMS-2023239 (TRIPODS).}\\
University of Wisconsin-Madison\\
{\tt zarifis@wisc.edu}\\
}

\begin{document}

\maketitle

\begin{abstract}%
We study the problem of estimating the mean of an identity covariance Gaussian in the 
truncated setting, in the regime when the truncation set comes from a low-complexity 
family $\cal{C}$ of sets. Specifically, for a fixed but unknown truncation set 
$S \subseteq \R^d$, we are given access to samples from the distribution 
$\cN(\bm \mu, \vec I)$ truncated to the set $S$. The goal is to estimate $\bm \mu$ 
within accuracy $\eps>0$ in $\ell_2$-norm. Our main result is a Statistical Query (SQ) 
lower bound suggesting a super-polynomial information-computation gap for this task. In 
more detail, we show that the complexity of any SQ algorithm for this problem is 
$d^{\poly(1/\eps)}$, even when the class $\cal{C}$ is simple so that $\poly(d/\eps)$ 
samples information-theoretically suffice. Concretely, our SQ lower bound applies when 
$\cal{C}$ is a union of a bounded number of rectangles whose VC dimension and Gaussian 
surface are small. As a corollary of our construction, it also 
follows that the complexity of the previously known algorithm for this task
is qualitatively best possible. 
\end{abstract}

\setcounter{page}{0}

\thispagestyle{empty}

\newpage

\section{Introduction}

We study the classical problem of high-dimensional statistical estimation from truncated (or 
censored) samples, with a focus on establishing {\em information-computation 
tradeoffs}. Truncation refers to the situation where samples falling outside of a 
fixed (potentially unknown) set are not observed. This phenomenon 
naturally arises in a wide range of applications across the sciences. 
Estimation from truncated samples has a rich 
history in statistics, dating back to \cite{bernoulli60}, who studied 
it in the context of smallpox vaccination. 
Pioneering early works include those of \cite{galton1898examination}, 
in the context of analyzing speeds of trotting horses; 
Pearson and Lee~\cite{pearson1902systematic,pearson1908generalised,lee1914table}, who used the method of moments for mean and standard deviation estimation 
from truncated Gaussian one-dimensional data; 
and~\cite{fisher1931properties}, who leveraged 
maximum likelihood estimation for the same problem. The reader is referred to \cite{schneider1986truncated,cohen1991truncated,balakrishnan2014art} for some textbooks on the topic.

Despite extensive investigation in the statistics community, the first 
statistically and computationally efficient algorithms for learning 
multivariate structured distributions in the truncated setting 
were developed fairly recently in the computer science community. 
The first such work \cite{daskalakis2018efficient} focuses on the fundamental setting of Gaussian mean and covariance estimation, and operates under the 
assumption that the truncation set is known (i.e., the learner is given oracle access to it). Most relevant to the current paper is the subsequent work 
of~\cite{kontonis2019efficient} that studies mean estimation of a spherical 
Gaussian under the assumption that the truncation set is unknown and is promised to 
lie in a family of sets with ``low complexity''.
Beyond mean and covariance estimation, a related line of work has addressed 
a range of other statistical tasks, including linear regression \cite{daskalakis2019computationally, daskalakis2020truncated, ilyas2020theoretical, daskalakis2021efficient, cherapanamjeri2023makes}, 
non-parametric estimation \cite{daskalakis2021statistical}, and learning other 
structured distribution families \cite{fotakis2020efficient, plevrakis2021learning, lee2023learning}.

In this paper, we focus on the basic task of estimating the mean of a spherical 
Gaussian in the truncated setting with unknown truncation set. 
The setup is as follows: 
Let $S \subseteq \R^d$ be a fixed subset of $\R^d$ 
and denote by $\alpha >0$ its probability mass under $\cN(\bm \mu, \vec I)$, 
the $d$-dimensional Gaussian with mean $\mu$ and identity covariance. 
Given access to samples from the distribution $\cN(\bm \mu, \vec I)$
{\em truncated to the set $S$}, the goal is to estimate $\bm \mu$ 
within accuracy $\eps>0$ in $\ell_2$ norm. 
For the special case of this task 
where the truncation set $S$ is \emph{known} to the algorithm 
(more accurately, the algorithm has oracle access to $S$), 
\cite{daskalakis2018efficient} gave a polynomial-time algorithm that draws 
$\tilde{O}_{\alpha}(d/\eps^2)$ truncated samples\footnote{The notation 
$\tilde{O}_{\alpha}$ suppresses poly-logarithmic dependence on its argument and some dependence on the parameter $\alpha$. In the context of the lower bounds established here, $\alpha$ will 
be a positive universal constant, specifically $\alpha>1/2$. }. They also pointed out 
that if $S$ is \emph{unknown}, and arbitrarily complex, 
then the learning problem is not solvable to better than constant accuracy, 
with any finite number of samples.

Although the latter statement %
might seem discouraging,
a natural avenue to circumvent this bottleneck is  restricting the set $S$ 
to a family of ``low complexity''. For example, early work in the statistics community 
considered the case where $S$ is a rectangle (box) or a union of a small number of 
rectangles. 
Intuitively, for such ``simple'' classes of sets, positive results may be attainable, 
even for unknown truncation set. 
\cite{kontonis2019efficient} formalized this intuition by providing 
the first positive results --- both information-theoretic and algorithmic --- 
for settings where the unknown set $S$ has ``low complexity''. 
Specifically, \cite{kontonis2019efficient} showed two (incomparable) 
positive results, corresponding to natural complexity measures 
of the family of sets containing $S$: 
\begin{enumerate}[leftmargin=*]
\item If $S$ comes from a family of sets $\cal{C}$ with VC-dimension $V$, 
then the problem is information-theoretically solvable to $\ell_2$ error $\eps$
with $\tilde{O}(V/\eps+d^2/\eps^2)$ truncated samples. 
\item If $S$ comes from a family of of sets $\cal{C}$ with 
Gaussian Surface Area at most $\Gamma>0$ (Definition~\ref{def:GSA}), 
then the problem is solvable using sample and computational complexity 
$d^{\Gamma^2 \poly(1/\eps)}$.
\end{enumerate}
For the setting of bounded VC-dimension, \cite{kontonis2019efficient} stated that 
``Obtaining a computationally efficient algorithm seems unlikely, unless one restricts 
attention to simple specific set families [...]''. For the setting of bounded surface 
area, the algorithm of \cite{kontonis2019efficient} has sample and computational 
complexity $d^{\poly(1/\eps)}$, even for $\Gamma = O(1)$, which is not required for 
simple classes of sets. This discussion serves as a natural motivation for the following question:
\begin{center}
    \emph{Are there ``simple'' families of sets for which learning truncated Gaussians\\ 
    exhibits an information-computation tradeoff?}
\end{center}
In more detail, is there a class of sets $\cal{C}$ such that our learning 
task is information-theoretically solvable with a few samples, 
and at the same time any {\em computationally efficient} 
algorithm requires significantly more samples?

We tackle this question in two well-studied restricted models of computation ---  
namely, in the Statistical Query (SQ) model~\cite{Kearns:98} 
and the low-degree polynomial testing model~\cite{hopkins2018statistical, kunisky2022notes}. 
{\em As our main result, we answer the above question in the affirmative 
for both of these models.} Specifically, we exhibit a family of sets with 
small VC dimension {\em and} small Gaussian surface area (hence, for which 
our problem is solvable with polynomial sample complexity), 
such that any SQ algorithm (and low-degree polynomial test) 
necessarily requires {\em super-polynomial} complexity.
As a corollary of our construction, it also follows that the complexity of the 
algorithm in \cite{kontonis2019efficient} (which is efficiently implementable in these models) is qualitatively best possible. 
Finally, we remark that the underlying family of sets used in our hard
instance is quite simple --- 
consisting of unions of a bounded number of rectangles.

\subsection{Our Results}

To formally state our main result, we summarize the basics of the SQ model.

\vspace{0.1cm}

\noindent{\bf SQ Model Basics}
The model, introduced by \cite{Kearns:98} and extensively studied since, 
see, e.g.,  \cite{FGR+13}, considers algorithms that, 
instead of drawing individual samples from the target distribution, 
have indirect access to the distribution using the following oracle:
\begin{definition}[STAT Oracle]
    Let $D$ be a distribution on $\R^d$. A statistical query is a bounded function $f: \R^d \to [-1,1]$. For $\tau > 0$, the $\mathrm{STAT}(\tau)$ oracle responds to the query $f$ with a value $v$ such that $|v - \E_{X \sim D}[f(X)]| \leq \tau$. We call $\tau$ the tolerance of the statistical query.
\end{definition}

\noindent An \emph{SQ lower bound} for a learning problem is an unconditional statement 
that any SQ algorithm for the problem either needs to perform a large number $q$ of 
queries, or at least one query with very small tolerance $\tau$. 
Note that, by Hoeffding-Chernoff bounds, a query of tolerance $\tau$ 
is implementable by non-SQ algorithms by drawing $O(1/\tau^2)$ 
samples and averaging them. Thus, an SQ lower bound intuitively 
serves as a tradeoff between runtime of $\Omega(q)$ and sample complexity of 
$\Omega(1/\tau)$.

\vspace{0.2cm}

\noindent{\bf Main Result} We are now ready to state our main result:

\begin{theorem}[SQ Lower Bound for Learning Truncated Gaussians]\label{thm:main-theorem}
   Let $d, k\in \Z_+$, $\eps>d^{-c} $ for some sufficiently small constant $c>0$, and assume $k \leq c/\eps^{0.15}$.
   Let $\cC$ be the class of all sets $S \subseteq \R^d$ with the properties that: (i) $S$ is the complement of a union of at most $k^2$ rectangles, and (ii) $S$ has $O(1)$ Gaussian surface area and $\Omega(1)$ mass under the target Gaussian. 
   Suppose that $\cA$ is an algorithm with the guarantee that, given SQ access to $\cN(\bm \mu, \vec I)$ truncated on a set $S \in \cC$ (where $\bm \mu$ and $S$ are unknown to the algorithm), it outputs a $\hat{\bm \mu}$ with $\| \hat{\bm \mu} - \bm \mu \|_2 \leq \eps$. Then, $\cA$ either performs $2^{d^{\Omega(1)}}$ many queries, or makes at least one query with tolerance $d^{-\Omega(k)}$.
\end{theorem}

We conclude with some remarks about our main theorem.
First, our SQ lower bound holds against a simple family of sets. 
The class $\cC$, as (the complement of) a union of $k^2 \leq \poly(1/\eps)$ 
rectangles, has VC dimension $\poly(d\log(k))$. By the sample upper bound of 
\cite{kontonis2019efficient}, the corresponding learning problem 
is thus solvable with $\poly(d/\eps)$ samples. 
Even for this simple class, our result suggests that any efficient SQ algorithm requires $d^{\poly(1/\eps)}$ samples.
Second, the fact that our family of sets has bounded Gaussian surface area implies that the algorithm of~\cite{kontonis2019efficient} (which fits in the SQ model)
is qualitatively optimal.
Finally, using a known equivalence between SQ and low-degree polynomial tests \cite{BBHLS20}, a qualitatively similar lower bound holds for the latter model. This implication is formally stated in~\Cref{sec:low_degree}.

\subsection{Overview of Techniques}\label{sec:techniques}

Our SQ lower bound leverages the methodology of \cite{DKS17-sq} and in 
particular the low-dimensional extension from \cite{diakonikolas2021optimality} (see also~\cite{DKRS23}), which provides a generic SQ 
hardness result for the problem of \emph{non-Gaussian component analysis}:
Fix a low-dimensional distribution $A$ on $\R^m$ with $m \ll d$, 
and consider the family $\cD$ of all $d$-dimensional distributions 
defined to  coincide with $A$ in some (hidden) $m$-dimensional subspace 
and being equal with the standard Gaussian in its orthogonal complement. 
The main result of that framework (cf. \Cref{lem:sq-hardness}) is that, if 
$A$ is itself similar to $\cN(\vec 0,\vec I_{m \times m})$ --- 
in the sense that it matches its first $k$ moments with 
$\cN(\vec 0,\vec I_{m \times m})$ --- 
then the hypothesis testing problem of distinguishing between a member of $\cD$ and $\cN(\vec 0,\vec I_{d \times d})$ requires either $2^{d^{\Omega(1)}}$ statistical queries, or a query of small tolerance $\tau < d^{-\Omega(k)}$. 
\new{This generic hardness result has been the basis of SQ lower bounds for a range of tasks, including} learning mixture models~\cite{DKS17-sq, DiakonikolasKPZ23, DKS23}, 
robust mean/covariance estimation~\cite{DKS17-sq}, robust linear regression~\cite{DKS19}, 
learning halfspaces and other natural concepts 
with label noise~\cite{DKZ20,  GoelGK20, DK20-Massart-hard, diakonikolas2021optimality, DKKTZ21-benign}, 
list-decodable estimation~\cite{DKS18-list, DKPPS21}, 
learning simple neural networks~\cite{DiakonikolasKKZ20}, 
and generative models~\cite{Chen0L22}. 

Given this generic SQ lower bound, we want to formulate our 
learning problem as a valid instance 
of non-Gaussian component analysis. That is, we aim to find a (low-
dimensional) distribution $A$ 
that matches $k=\poly(1/\eps)$ moments with 
$\cN(\vec 0,\vec I_{m \times m})$ and is i{\em tself a truncated Gaussian} 
with mean $\bm {\mu}$ where $\| \bm {\mu} \|_2 \geq \eps$, truncated 
on a set $S$ of large mass with small Gaussian surface area \new{and VC dimension}. This 
would imply that learning the mean of truncated Gaussians within 
error $\eps$ in $d$ dimensions \new{has SQ complexity $d^{\poly(1/\eps)}$}.

A first attempt is to try to find a one-dimensional distribution $A$ 
for the above construction, in particular an $A$ of the form $\cN(\eps,1)$ 
conditioned on a set $S$ which is a union of a small number of intervals. 
We start by noting that it suffices to make this construction work 
for {\em any} finite number 
of intervals --- indeed, an existing technique from \cite{DKZ20} 
can be leveraged to show that if $k$ moments can be matched 
using a finite union of intervals, they can also be matched using just $k$ 
intervals (\Cref{prop:main_stuct}). Without having to worry about the number 
of intervals, our proof strategy would be as follows:
The first step is to create a continuous version of the construction. 
Namely, we wish to find a function $f: \R \to [0,1]$ so that 
if the probability density function of $\cN(\eps,1)$ is multiplied by $f$ 
and re-normalized, we obtain a probability distribution that matches $k$ 
moments with $\cN(0,1)$ (here $f$ represents some fractional version of the 
indicator function of $S$). This can be done somewhat explicitly. 
In particular, we can take $f$ to be a carefully chosen exponential 
function, so that the density of $\cN(\mu,1)$ times $f$ re-normalized is 
exactly $\cN(0,1)$ (cf. \Cref{cl:explicit_calc}). Unfortunately, 
this $f$ will not be bounded in $[0,1]$, 
and in particular in the extreme tails 
will have $f(x) > 1$. However, since so little mass lies at these 
tails, if we truncate $f$ to have value at most $1$, 
we do not change the first $k$ moments by much (cf. \Cref{cl:moment_negligible}). 
Then, using a technique of \cite{DKS17-sq} (also see Chapter 8 in \cite{diakonikolas2023algorithmic}), we can modify $f$ slightly 
(by adding a carefully chosen polynomial times the indicator function of an interval) to fix this moment discrepancy.

The above sketch gives a one-dimensional construction, where $S$ is a 
union of at most $k$ intervals. Unfortunately, this class of sets will 
have surface area approximately $k$, which is far too large for our purposes. 
In fact, for any reasonable one-dimensional set $S$, 
we will expect to have at least constant surface area 
(as a single point on the boundary of $S$ contributes this much). 
To overcome this obstacle, we will need to consider 
a two-dimensional construction instead 
(eventually given in \Cref{lem:A}). That is, we want to exhibit 
a family of sets $S \subseteq \R^2$ so that if the two-dimensional 
Gaussian $\cN(( \eps,0), \vec I_{2\times 2})$ is conditioned on $S$, 
we match $k$ low-degree moments with $\cN(\vec 0,\vec I_{2\times 2})$. 
Specifically, we will take $S$ to be the complement of an appropriate union 
of rectangles in $\R^2$ (cf.\xspace \Cref{fig:rectangles}). 
We first describe the goal of our construction for each axis separately:
For the $y$-axis, we need to find a small union of intervals $U$ 
such that (i) the mass of $U$ is $\delta_1 = \poly(\eps)$, 
and (ii) $\cN(0,1)$ conditioned outside of $U$ matches $k$ moments with $\cN(0,1)$.
For the $x$-axis, we need another union of intervals $T$, 
which also has small mass $\delta_2 = \poly(\eps)$, 
and such that the pdf of $\cN(\eps,1)$ multiplied by 
$(1-\delta \1(x \in T))$ matches its first $k$ moments with $\cN(0,1)$. 
The multiplication by $(1-\delta \1(x \in T))$ is needed to take into account the $\delta$-mass removed in the $y$-axis earlier. 
After having these at hand, we can let $S$ be the complement of $T \times U$. 
By a direct computation one can show that, given the properties above, 
$\cN((\eps,0),\vec I_{2 \times 2})$ conditioned on $S$ matches its low-degree  
moments with $\cN(\vec 0,\vec I_{2\times 2})$ 
(cf. calculation in \eqref{eq:momentmatch2d}). 
We note that the boundary of $S$ consists of $k^2$ rectangles, 
each with perimeter approximately $\delta_1+\delta_2 = \poly(\eps)$. 
So, if $k \ll 1/\sqrt{\delta_1+\delta_2}$, $S$ will have small Gaussian surface area.
Finally, the plan for showing existence of the sets $U$ and $T$ 
is the following: To establish the existence of the $U$ set 
(cf.\xspace\Cref{prop:matching1}), we first provide an explicit construction of 
intervals, by splitting the real line into tiny intervals and defining $U$ to 
include $1-\delta$ fraction of each; and then leverage 
the technique from \cite{DKZ20} to reduce the number of intervals down to 
$k$. The proof strategy for the set $T$ (\Cref{prop:matching2}) is essentially the one 
that was discussed in the previous paragraph.

\begin{figure}[h]
    \centering
    \includegraphics[width=0.57\textwidth]{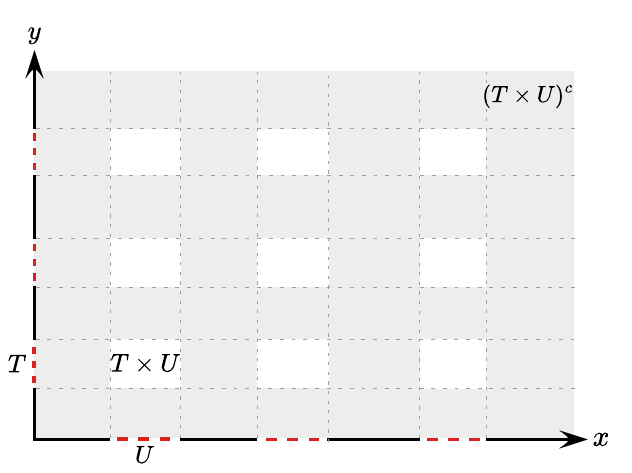} %
    \caption{Truncation set in $\mathbb{R}^2$. The red dotted parts of the horizontal and vertical axes represent the unions of intervals $U$ and $T$, respectively. The white rectangles represent the set $T \times U$, and the remaining gray area of $\mathbb{R}^2$ is their complement, denoted as $(T \times U)^c$, on which we truncate the Gaussian distribution $\mathcal{N}((\epsilon,0),\mathbf{I}_{2\times 2})$.}
    \label{fig:rectangles}
\end{figure}

\section{Preliminaries}\label{sec:prelims}

\noindent{\bf Basic Notation}
We use the notation $[n] \eqdef \{1,\ldots,n\}$.
 We use $\Z$ for positive integers, and $\|\cdot\|_2$ for the $\ell_2$-norm of vectors. 
	 We use $\cN(\bm \mu, \vec \Sigma)$ to denote the Gaussian with mean $\bm \mu$ and covariance matrix $\vec \Sigma$ and use $\phi_{\bm \mu, \vec \Sigma}(\vec x)$ for its probability density function. For other distributions, we will slightly abuse notation by using the same letter for a distribution and its pdf, e.g., we will denote by $P(\bx)$ the pdf of a distribution $P$.

\begin{definition}[Truncated Gaussian]
For a set $S \subseteq \R^d$, a vector $\bm \mu \in \R^d$ and a PSD matrix $\vec \Sigma \in \R^{d \times d}$, we define $\cN(\bm \mu, \vec \Sigma, S)$ to be the Gaussian with mean $\cN(\bm \mu, \vec \Sigma)$ after truncation using the set $S$, i.e., the distribution with the following pdf (where $\phi_{\bm \mu, \vec \Sigma}$ denotes the pdf of $\cN(\bm \mu, \vec \Sigma)$): $\phi_{\bm \mu, \vec \Sigma,S}(x) :=  Z^{-1} \1(x \in S) \phi_{\bm \mu, \vec \Sigma}(\vec x)$, where $Z:= \int_{\R^d} \1(x \in S) \phi_{\bm \mu, \vec \Sigma}(\vec x) \d \vec x$.
\end{definition}

We now define Gaussian Surface Area (GSA), which has served as a complexity measure of sets in learning theory 
and related fields; see, e.g.,~\cite{KOS:08, Kane11, Neeman14}.

\begin{definition}[Gaussian Surface Area]
\label{def:GSA}
  For a Borel set $A \subseteq \R^d$, its Gaussian surface area is defined by
  $
  \Gamma(A) \eqdef \liminf_{\delta \to 0} \frac{\normal(A_\delta \setminus A)}{\delta},
  $  where $A_\delta = \{x : \mathrm{dist}(x, A) \leq \delta\}$.
\end{definition}

\noindent{\bf Additional Background on the SQ Model} 
The main fact that we use from the SQ literature \cite{FGR+13,DKS17-sq} concerns the family of distributions which are standard Gaussian along every direction, except from a low-dimensional subspace, where they are forced to be equal to some other (non-Gaussian) distribution $A$.

\begin{definition}[Hidden Direction Distribution]\label{def:hidden}
    For an $m$-dimensional distribution $A$ and a matrix $\vec V \in \R^{m \times d}$, we define the distribution $P_{A, \vec  V}$ with pdf $A(\vec V \bx) (2\pi)^{-\frac{(d-m)}{2}}  e^{-\frac{1}{2}\|\bx - \vec V^\top \vec V \bx\|_2^2}$.
\end{definition}
The main result from \cite{DKS17-sq} is that, if $A$ is similar to Gaussian, in the sense that its first moments agree with those of $\cN(\vec 0, \vec I)$, then the hypothesis testing problem between $\cN(\vec 0, \vec I)$ and a distribution of the above family is hard for any SQ algorithm. The following fact shows formally this hardness. See \Cref{app:sq} for related preliminaries and the proof of the fact below; $\chi^2(A,B)$ below is defined as $\int_{\R^d} A^2(\bx)/B(\bx)\, \d\bx - 1$. 
\begin{fact}\label{lem:sq-hardness}
Let $d,k\in \Z$ and $m<d^{1/10}$ and $k<d^c$ for some sufficiently small constant $c>0$.  Let  $A$ be a distribution over $\R^m$ such that its first $k$ moments 
match the corresponding moments of $\normal(\vec 0,\vec I_m)$. Define the family $\D$ of distributions containing 
$P_{A, \vec U}$ (cf. \Cref{def:hidden}) for all matrices $\vec U\in \R^{m\times d}$ such that $\vec U \vec U^\top =\vec I_m$.
Then, any SQ algorithm that distinguishes between $\mathcal{N}(\vec 0,\vec I_d)$ and  $\D$ requires either $2^{d^{\Omega(1)}}$ many queries, or at least one query with tolerance $d^{-\Omega(k)} \sqrt{\chi^2(A,\mathcal{N}(\vec 0,\vec I_m))}$.
\end{fact}

\section{SQ Lower Bound For Truncated Gaussians}\label{sec:section 3}

In this section we formalize the proof strategy for \Cref{thm:main-theorem} which had been informally described in \Cref{sec:techniques}. We will show a stronger version of that theorem, stated below, which concerns hypothesis testing between the standard Gaussian and a truncated Gaussian.

\begin{theorem}[SQ Lower Bound; Hypothesis Testing Hardness]\label{thm:main-theorem-hyp} 
Let $d, k \in \Z_+$, $\eps>d^{-c} $ for some sufficiently small constant $c>0$, and $k \leq c/\eps^{0.15}$. 
Let $\cC$ be the class of all sets $S \subseteq \R^d$ with the properties that: (i) $S$ is a union of at most $k^2$-many $d$-dimensional rectangles, and (ii) $S$ has $O(1)$ Gaussian surface area and $\Omega(1)$ mass under the target Gaussian. 
    Consider the hypothesis testing problem defined below: 
    \begin{enumerate}[leftmargin=*]
        \item Null Hypothesis: $D=\cN(\vec 0,\vec I)$.
        \item Alternative Hypothesis: $P \in \cD$, where $\cD$ is the family of truncated Gaussians $\cN(\bm \mu, \vec I, S)$ for all $\| \bm \mu \|_2 \geq \eps$ and $S \in \cC$.
    \end{enumerate}
    Then, any SQ algorithm that solves the above problem, either performs $2^{d^{\Omega(1)}}$ many queries or performs at least one query with tolerance $d^{-\Omega(k)}$.
\end{theorem}

Note that \Cref{thm:main-theorem-hyp} implies immediately \Cref{thm:main-theorem} by a straightforward reduction: One can first find $\bm \mu$ approximating the true mean up to error $\eps/2$ and then reject the null hypothesis if $\|\bm \mu\|>\eps/2$.

The end goal towards showing \Cref{thm:main-theorem-hyp} is to establish the existence of the following two-dimensional truncated Gaussian distribution $A$ that matches $k = \poly(1/\eps)$ moments with the standard Gaussian (\Cref{lem:A}).

\begin{restatable}{proposition}{FINALPROPOSITION}\label{lem:A}
Let $c>0$ be a sufficiently small absolute constant, $\eps\in (0,c)$, and $k = c/\eps^{0.15}$.
There exists a distribution $A$ on $\R^2$,  for which the following are true: 
\begin{enumerate}[leftmargin=*]
    \item\label{it:momentmatch} $A$ matches its first $k$ moments with the 2-dimensional standard Gaussian. 
    \item $\chi^2(A,\cN(\vec 0,\vec I_2)) = O(1)$.
    \item \label{it:valid_instance} Every distribution of the form $P_{A,\vec V}$ (cf. \Cref{def:hidden}) can be written as a truncated Gaussian $\cN(\bm \mu, \vec \Sigma, S)$ for $\bm \mu = (\eps,0)$, $\vec \Sigma = \vec I_2$ (the $2\times 2$ identity matrix) and some $S \subseteq \R^d$ which has mass (with respect to the target Gaussian) at least $1/2$ and Gaussian surface area at most $1$.\looseness=-1
\end{enumerate}
\end{restatable}

\noindent Having \Cref{lem:A} at hand, then \Cref{thm:main-theorem} follows from \Cref{lem:sq-hardness}.

The proof of \Cref{lem:A} requires two key results (\Cref{prop:matching1,prop:matching2}). 
In \Cref{lem:A}, $A$ is a 2-dimensional distribution that matches moments with the standard normal. In the following lemmata, we construct independently each dimension of that distribution. 
The marginal on the $y$-axis will be a standard normal, conditioned on a union of $k$ intervals, as shown in \Cref{prop:matching1} below.  
As mentioned in the proof sketch of \Cref{sec:techniques}, we want these intervals to have small mass, thus we will eventually use $\delta=\sqrt{\eps}$ below.
We defer the proof to \Cref{sec4}.
\begin{lemma}\label{prop:matching1}
    For any $\delta \in (0,1)$ and $k\in \Z_+$ there exists a set $U \subseteq \R$ such that: $U$ is a union of $k$ intervals with $\Pr_{y \sim \cN(0,1)}[y \in U] = \delta$ and for all $t=1,\ldots,k$ it holds
    \begin{align}\label{eq:marginal-y}
        \E_{y \sim \cN(0,1)}[y^t \; | \; y \not\in U] = \E_{y \sim \cN(0,1)}[y^t]\;.
    \end{align}
\end{lemma}

Next we construct the marginal of $A$ for the $x$-axis. In this case, we start with a Gaussian distribution with mean $\eps$, and we reweight it with a $k$-piecewise constant function taking values in $\{1-\delta,1 \}$ so that it matches $k$ moments with the standard normal. The reason why we use values in $\{1-\delta,1 \}$ is because we have removed $\delta$ mass in our construction for the $y$-axis.  This will be clearer when we provide the calculation that $A$ matches moments with the 2-dimensional Gaussian.
The proof can be found on \Cref{sec5}.
\begin{lemma}\label{prop:matching2}
    Let $c>0$ be a sufficiently small absolute constant and $\eps\in (0,c)$. Let $\delta,k$ be parameters so that $\delta = \sqrt{\eps}$ and $k = c/\eps^{0.15}$.
    There exists a set $T\subseteq \R$ such that: $T$ is a union of $k$ intervals, 
    $\Pr_{x \sim \cN(\eps,1)}[x \in T] \leq \eps^{0.3}$, and for all $t=0,\ldots,k$ it holds
    \begin{align}\label{eq:marginal-x}
        \E_{x \sim \cN(\eps,1)}[x^t(1-\delta\1\{x\in T\})] Z^{-1}
        = \E_{x \sim \cN(0,1)}[x^t] \;,
    \end{align}
    where $Z= \E_{x \sim \cN(\eps,1)}[1-\delta\1\{x\in T\}]$.
\end{lemma}

Using \Cref{prop:matching1,prop:matching2}, we can prove \Cref{lem:A} by letting 
\begin{align*}
    A(x,y) = \frac{\phi(x-\eps) \phi(y) \1( (x,y) \not\in T \times U)}{Z} \;.
\end{align*}
In particular, it can be seen that $A$ matches $k$ moments with the $2$-dimensional normal by a direct computation that uses \eqref{eq:marginal-y}, \eqref{eq:marginal-x}. We provide the proof below.
\begin{proof}[Proof of \Cref{lem:A}]
Let $T,U,\delta,Z$ as in \Cref{prop:matching1,prop:matching2},
we let $A$ to be a distribution defined by the following probability density function:
\begin{align}
    A(x,y) = \frac{\phi(x-\eps) \phi(y) \1( (x,y) \not\in T \times U)}{Z} \;.
\end{align} 

We start with \Cref{it:momentmatch}. First we note that $A$ is indeed a valid distribution, i.e., the normalizing factor is correct.
\begin{align*}
    Z &= \int\limits_{-\infty}^{+\infty}\int\limits_{-\infty}^{+\infty} D(x,y)\d y \,\d x 
    = \int\limits_{-\infty}^{+\infty} \phi(x-\eps)\left(\1(x \in T) \int_{y \not\in U}\phi(y) \d y  + \1(x \not\in T) \int\limits_{-\infty}^{+\infty}\phi(y) \d y  \right) \d x \\
    &= \int\limits_{-\infty}^{+\infty} \phi(x-\eps)\left(\1(x \in T)(1-\delta) +  \1(x \not\in T)  \right) \d x \tag{by \Cref{prop:matching1}}\\
    &= \int\limits_{-\infty}^{+\infty} \phi(x-\eps)(1-\delta \1(x \in T)) \d x \;,
\end{align*}
where the calculation essentially used the geometry of our sets (see also \Cref{fig:rectangles}): For any fixed $x$, there are two cases. If $x \not\in T$, then no Gaussian mass is removed from the $y$-integral; otherwise, a $(1-\delta)$ mass is removed (as explained in \Cref{prop:matching1}).

We can similarly see that $A$ matches the first $k$ moments with the standard two dimensional Gaussian: Let $t$ and $s$ be non-negative integers with $t+s \leq k$. Then,
\begin{align}
    \frac{1}{Z} &\int\limits_{-\infty}^{+\infty} \int\limits_{-\infty}^{+\infty} x^t y^s \phi(x-\eps)\phi(y)\1( (x,y) \not\in T \times U)  \d y \, \d x \label{eq:momentmatch2d}\\
    &= \frac{1}{Z} \int\limits_{-\infty}^{+\infty}  x^t  \phi(x-\eps) \left( \1(x \in T) \int\limits_{y \not\in U}  y^s \phi(y) \d y +  \1(x \not\in T)   \int\limits_{-\infty}^{+\infty} y^s \phi(y) \d y \right) \d x \\
    &=\frac{1}{Z} \int\limits_{-\infty}^{+\infty}  x^t  \phi(x-\eps)   \int\limits_{-\infty}^{+\infty}  y^s \phi(y)\left( \1(x \in T)(1-\delta) + \1(x \not\in T) \right) \d y    \d x \tag{by \Cref{prop:matching1}}\\
    &=   \int\limits_{-\infty}^{+\infty}  x^t  \frac{\phi(x-\eps)   (1-\delta \1(x \in T))}{Z} \d x \int\limits_{-\infty}^{+\infty} y^s \phi(y)\d y   \\
    &= \int\limits_{-\infty}^{+\infty}  x^t \phi(x) \d x \int\limits_{-\infty}^{+\infty}  y^s \phi(y) \d y \;. \tag{by \Cref{prop:matching2}}
\end{align}

Finally, it is easy to see that the chi-square of $A$ is $O(1)$.
    \begin{align*}
        \chi^2(A,\cN(\vec 0,\vec I_2)) &= \int\limits_{-\infty}^{+\infty} \int\limits_{-\infty}^{+\infty} \frac{A^2(x,y)}{\phi(x)\phi(y)}\d y \d x - 1 \\
        &=  \int\limits_{-\infty}^{+\infty}\frac{\phi^2(x-\eps)}{\phi(x)} \int\limits_{-\infty}^{+\infty} \phi(y) \1((x,y) \not\in T \times U) \d y \d x - 1 \\
        &\leq \int\limits_{-\infty}^{+\infty}\frac{\phi^2(x-\eps)}{\phi(x)} \d x = \chi^2(\cN(\eps,1),\cN(0,1)) = e^{\eps^2} \;.
    \end{align*}

    We move to \Cref{it:valid_instance}. The fact that $P_{A,\vec V}$ is a truncated Gaussian follows trivially by our definitions. The Gaussian surface area bound comes from the fact that $T \times U$ is a union of at most $k^2$ many rectangles, each with perimeter $O(\delta )$ (this is because the sets $T$ and $U$ from \Cref{prop:matching1,prop:matching2} have mass at most $O(\delta)$). Using $\delta = \sqrt{\eps}$ and $k\ll 1/\eps^{1/4}$, we obtain that the Gaussian surface area is at most $1$.
\end{proof}

\section{Proof of \Cref{prop:matching1}}
\label{sec4}

Regarding \Cref{prop:matching1}, we need to find a union $U$ of $k$ intervals such that the truncated version of $\cN(0,1)$ on these intervals matches moments with $\cN(0,1)$, and the mass of $U$ under $\cN(0,1)$ is equal to a parameter $\delta$ of our choice. 
The proof strategy is the following:
First, we note in \Cref{lem:equivalence} that it suffices to find a piecewise constant function $f : \R \to \{-\delta, 1-\delta \}$ such that $\E_{z \sim \cN(0,1)}[z^t f(z)]=0$, i.e., the weighted by $f$ moments of $\cN(0,1)$ are zero. \Cref{lem:equivalence} implies that once such a function $f$ is found, \Cref{prop:matching1} follows by letting the set $U$ be the union of all the intervals where $f(z)>0$. We proceed to showing the existence of $f$ through a two-step process. We start with an explicit construction in \Cref{lem:function}. Although capable of making the weighted moments arbitrarly close to zero, this construction yields a function with a significantly larger number of pieces than $k$. We are then able to reduce the number of pieces down to the desired count of $k$  using a technique from \cite{DKZ20}, implemented in \Cref{prop:main_stuct}.

\begin{restatable}{claim}{EQUIVALENCE}\label{lem:equivalence}
    Let $U \subseteq \R$ and $t \in \Z_+$ be a set and an integer. Define the piecewise constant function
    \begin{align*}
        f(z) = 
        \begin{cases}
            1-\delta  , &z \in U \\
            -\delta  , &z \not\in U \;,
        \end{cases}
    \end{align*}
    with $\delta:= \Pr_{z \sim \cN(0,1)}[z \in U]$.
    The following three statements are equivalent:
    \begin{enumerate}
        \item $\E_{z \sim \cN(0,1)}[z^t f(z)] = 0$.\label{it:matching1}
        \item $\E_{z \sim \cN(0,1)}[z^t | z \in U] = \E_{z \sim \cN(0,1)}[z^t ]$. \label{it:matching2}
        \item $\E_{z \sim \cN(0,1)}[z^t | z \not\in U] = \E_{z \sim \cN(0,1)}[z^t ]$.\label{it:matching3}
    \end{enumerate}
\end{restatable}

\begin{proof}
    We first show the equivalence between \Cref{it:matching2} and \Cref{it:matching3}. We assume \Cref{it:matching2} and show \Cref{it:matching3} (the other direction is identical):
    \begin{align}
        \E_{z \sim \cN(0,1)}[z^t] &= 
        \E_{z \sim \cN(0,1)}[z^t \1(z \in U)] + \E_{z \sim \cN(0,1)}[z^t \1(z \not\in U)]\notag\\
        &= \delta \E_{z \sim \cN(0,1)}[z^t | z \in U] + (1-\delta) \E_{z \sim \cN(0,1)}[z^t | z \not\in U] \label{eq:decompose}\\
        &= \delta \E_{z \sim \cN(0,1)}[z^t] + (1-\delta) \E_{z \sim \cN(0,1)}[z^t | z \not\in U] \;, \notag
    \end{align}
    where the last line used \Cref{it:matching2}. Rearranging, this means that $\E_{z \sim \cN(0,1)}[z^t | z \not\in U] = \E_{z \sim \cN(0,1)}[z^t]$.

    We now prove that \Cref{it:matching2} and \Cref{it:matching3} imply \Cref{it:matching1}, i.e., $\E_{z \sim \cN(0,1)}[z^t f(z)] = 0$.
    \begin{align}
        \frac{1}{\delta(1-\delta)}\E_{z \sim \cN(0,1)}[z^t f(z)] &= \frac{1}{\delta}\E_{z \sim \cN(0,1)}[z^t \1(z \in U)] - \frac{1}{1-\delta} \E_{z \sim \cN(0,1)}[z^t \1(z \not\in U)] \notag \\
        &= \E_{z \sim \cN(0,1)}[z^t | z \in U] - \E_{z \sim \cN(0,1)}[z^t | z  \not\in U] = 0 \;, \label{eq:expzero}
    \end{align}
    where the last equality is due to the part we showed earlier.

    The direction from \Cref{it:matching1} to \Cref{it:matching2} is similar: By writing out $\E_{z \sim \cN(0,1)}[z^t f(z)] = 0$ similarly to \eqref{eq:expzero} we can see that $\E_{z \sim \cN(0,1)}[z^t | z \in U]=\E_{z \sim \cN(0,1)}[z^t | z \not\in U]$. Then, using this in \eqref{eq:decompose} we get that both conditional expectations are equal to $\E_{z \sim \cN(0,1)}[z^t]$.
\end{proof}

Next we explicitly construct a piecewise constant function from $\R$ to 
$\{-\delta,1-\delta \}$ that achieves zero weighted moments (or, more accurately, arbitrarily small weighted moments).

\begin{restatable}{claim}{EXPLICITFUNC}\label{lem:function}
For any $\eta>0$ and $\delta \in (0,1)$, there exists a $(k \log(1/\eta))^{2k}/\eta$-piecewise constant function
$g_{\eta}:\R\to\{1-\delta,-\delta\}$ such that $|\E_{z \sim \normal(0,1)}[g_{\eta}(z)z^t]|\leq \eta$, for all $t=0,1,\ldots,k$ and $\Pr_{z \sim \cN(0,1)}[f(z)= 1-\delta] \in [\delta-\eta,\delta+\eta]$.
\end{restatable}

We sketch the proof of \Cref{lem:function}, with the full version being deferred to  \Cref{app:explicit}.
The idea is that we partition the real line into intervals $A_i = [is,(i+1)s]$ for $i \in \Z$ using a small step size $s$. For each $i \in \Z$, we further split $A_i$ into two parts $A_i^{+}:=[is,(i+\delta)s]$ and $A_i^{-}:= [(i+\delta)s,(i+1)s]$, i.e., the ratio of the sub-intervals' length is $\delta/(1-\delta)$. We define $g_\eta(z) = 1-\delta$ on  $A_i^{+}$ and $g_\eta(z) = -\delta$ on  $A_i^{-}$ for all $i$. The main argument is that since the Gaussian density does not change by much inside $A_i$, the contribution to the moment integral from the sub-intervals $A_i^{+}$ and $A_{i}^{-}$ must approximately adhere to the ratio of the  sub-intervals' lengths, i.e.,
\begin{align}
    \frac{\int_{A_i^{+}} z^t \phi(z) \d z }{\int_{A_i^{-}}  z^t \phi(z)\d z } = \frac{\delta}{1-\delta}(1 + \xi_i) \label{eq:ratiooo}
\end{align}
for some small $\xi_i$. In fact, we can show it for $|\xi_i| =O(i s^2)$ using a polynomial approximation for the density function $\phi(z)$. The important part is that we can control $|\xi_i|$ using the step size $s$. Finally, 
\begin{align}
    \left| \E_{z \sim \normal(0,1)}[g_\eta(z)z^t] \right| &= \left| \sum_{i \in \Z} \int_{z \in A_i} g_\eta(z)z^t \phi(z)\d z \right|
    \leq  \sum_{i \in \Z}\left| (1-\delta)\int_{z \in A_i^{+}} \hspace{-12pt} z^t \phi(z)\d z - \delta \int_{z \in A_i^{-}}\hspace{-12pt} z^t \phi(z) \d z   \right| \notag\\
    &\leq \sum_{i \in \Z} \left| \xi_i \, \delta \int_{z \in A_i^{-}} f(z)z^t \phi(z) \d z \right| < \eta \;. \tag{using \eqref{eq:ratiooo}}
\end{align}
The last step above amounts to a sufficiently small $s$  so that $\xi_i$ becomes sufficiently small and makes the entire right hand side  less than $\eta$. There are additional details needed to formalize this, such as noting that the summation does not need to cover the entire range of $i \in \mathbb{Z}$. We defer these details to \Cref{app:explicit}.

The final step is to reduce the number of pieces from $(k \log(1/\eta))^{2k}/\eta$ down to $k$. To this end, we use the proposition below which shows that we can start with a $t>k$ piecewise constant function and decrease the number of pieces to $k$ without changing the desired properties of the function. An analogous statement was shown 
in \cite{DKZ20}; here we require a generalization of this for all continuous distributions and any sequence of moments. 
The main idea of the proof is to model the endpoints of the intervals 
as a differential equation. To do so, we start with an instance 
that has many more endpoints than our goal, 
i.e., the instance has $t$ endpoints, 
and the first $k$ moments of this distribution have specific values. 
One can model this as a vector-valued function 
$\vec M(z_1,\ldots,z_t):\R^t\mapsto\R^{k}$, 
where $z_1,\ldots,z_t$ are the endpoints and $\vec M_i$ is the value 
of the $i$-th moment. 
Our task is to move the endpoints $z_i$ until two of them coincide 
or one of them goes to infinity, while keeping the vector $\vec M$ 
constant (so that the moments will continue to satisfy our assumptions). 
This is achieved by finding a specific $\vec u(z):\R\mapsto\R^t$ 
with the properties that $\vec u(0)=[z_1,\ldots,z_t]$ 
(so that the initial conditions satisfy our moment assumptions), 
$\d \vec M(\vec u_1(z),\ldots,\vec u_t(z))/\d z=\vec 0$ 
(so that the moments remain constant), 
and $\d \vec u_{t}(z)/\d z=1$ (so that at least one endpoint 
will be removed, i.e., in the worst case the $t$-th endpoint goes to infinity). One can show that such a function $\vec u$ always exists, 
as long as $t>k+1$. For completeness, we provide a proof in Appendix \ref{app:reduce}.

\begin{restatable}{proposition}{REDUCEINTERVALS}
    \label{prop:main_stuct}
Let $k,\ell$ be  positive integers with $\ell\geq k+1$ and $a,b\in \R$ with $b>a$. Let $D$ be a continuous distribution over $\R$ and let $\nu_0,\ldots,\nu_{k-1}\in\R$. If for any $\eta>0$ there exists an at most $\ell$-piecewise constant function $g_{\eta}: \R \to \{a,b\}$
such that $|\E_{z \sim D}[g_{\eta}(z)z^t]-\nu_t|\leq\eta$ for every non-negative integer $t<k$, then there exists an at most $(k+1)$-piecewise constant function $f: \R \to \{a,b\}$
such that $\E_{z \sim D}[f(z)z^t]=\nu_t$, for every non-negative integer $t<k$.
\end{restatable}
Having the above at hand, the proof of \Cref{prop:matching1} follows from \Cref{lem:function} and \Cref{prop:main_stuct} applied with $a = -\delta$, $b = 1+\delta$,
$D = \cN(0,1)$. The set $U$ that satisfies the conclusion of \Cref{prop:matching1} is the set of intervals on which $f(z)>0$.

\subsection{Proof of \Cref{lem:function}}\label{app:explicit}

Let $C$ be a sufficiently large absolute constant.
    Fix the parameters $s:=0.01\eta/(k \log(1/\eta))^k$,  $i_{\max} = 10\log^k(1/\eta)/s$ throughout the proof.
    We also define $U^+$ and $U^{-}$ to be the  unions of intervals in the positive and negative part of the real line as shown below. We define them so that their union $U= U^{+} \cup U^{-}$ is symmetric around zero:
    \begin{align*}
        U^+ &:= \left(\bigcup_{i=0}^{i_{\max}-1} [is, (i + \delta)s] \right) \cup [i_{\max} s,+ \infty)\;,\\
        U^{-} &:= \left(\bigcup_{i=0}^{i_{\max}-1} [- (i+\delta)s, -is]\right) \cup [ -\infty ,- i_{\max} s) \;.
    \end{align*} 
    Finally define the piecewise constant function
    \begin{align*}
        g_{\eta}(z) := 
        \begin{cases}
            1-\delta   &, z \in U \\
            -\delta   &, z \not\in U \;.
        \end{cases}
    \end{align*}

    First, we note that because of symmetry of $g_{\eta}(z)$ around zero
    \begin{align}
        \left|\E_{z \sim \normal(0,1)}[g_{\eta}(z)z^t] \right| \leq  \left| \int_{-\infty}^0 g_{\eta}(z) z^t \phi(z) \d z  \right| + \left| \int_{0}^{+\infty} g_{\eta}(z) z^t \phi(z) \d z  \right|
        = 2 \left| \int_{0}^{+\infty} g_{\eta}(z) z^t \phi(z) \d z \right| \;. \label{eq:maingoal}
    \end{align}
    Therefore, in everything that follows, it  suffices to only consider the integral on the positive part of the real line.

    Our goal is to bound \eqref{eq:maingoal} by  $\eta$. As a first step, we need the following bound on the ratio 
    of consecutive pieces of the moment integral:

     \begin{align}
        \frac{\int_{is}^{(i+\delta)s} z^t \phi(z) \d z }{\int_{(i+\delta)s}^{(i+1)s} z^t \phi(z)\d z } \leq \frac{\delta s ((i+\delta)s)^t \phi(i s)}{(1-\delta) s ((i+\delta)s)^t \phi((i+1) s)}
        \leq \frac{\delta}{1-\delta} e^{i s^2 + s^2/2} \leq \frac{\delta}{1-\delta}(1 + 2 i s^2) \;, \label{eq:ratio}
    \end{align}
    where we used the minimum and maximum values that the $z^t\phi(z)$ takes in each integral, and then used that $1+x \leq e^x \leq 1 + 1.1x$ for all $x<0.1$, where we applied this with $x = i s^2\leq i_{\max} s^2$ which is indeed less than $0.1$ for our choice of $s,i_{\max}$.

    We can now proceed to bound \eqref{eq:maingoal}. We start with the upper bound; see below for step by step explanations:
    \begin{align}
        \int_{0}^{+\infty} &g_{\eta}(z) z^t \phi(z) \d z \notag\\
        &= \int_{i_{\max} s}^{+\infty}z^t \phi(z) \d z + \sum_{i = 0}^{i_{\max} - 1} \int_{is}^{(i+1)s} g_{\eta}(z) z^t \phi(z) \d z    \notag\\
        &\leq \frac{\eta}{4} +\sum_{i = 0}^{i_{\max} - 1}\left\{ (1-\delta) \int_{is}^{(i+\delta)s} z^t \phi(z)\d z  - \delta\int_{(i+\delta)s}^{(i+1)s} z^t \phi(z) \d z \right\}   \label{eq:a2}\\
        &=\frac{\eta}{4} +\sum_{i = 0}^{i_{\max} - 1} \left\{ (1-\delta) \frac{\delta}{1-\delta}(1+2i s^2)\int_{(i+\delta)s}^{(i+1)s} z^t \phi(z) \d z   -\delta \int_{(i+\delta)s}^{(i+1)s} z^t \phi(z)\d z   \right\}  \label{eq:a3}\\
        &= \frac{\eta}{4} + \sum_{i = 0}^{i_{\max} - 1}\left\{ 2i \delta s^2\int_{(i+\delta)s}^{(i+1)s} z^t \phi(z)\d z   \right\}  \notag\\
        & \leq \frac{\eta}{4} + 20 s\log^k(1/\eta)   \sum_{i =0 }^{i_{\max}-1} \int_{(i+\delta)s}^{(i+1)s} z^t \phi(z)\d z     \label{eq:a5}\\
        & \leq \frac{\eta}{4} +20 \delta s\log^k(1/\eta) \int_{0}^{+\infty} z^t \phi(z)\d z  \label{eq:a6}\\
        &\leq \frac{\eta}{4} + 20s\log^k(1/\eta) (t-1)!!   \label{eq:a7}\\
        &\leq  \frac{\eta}{4} + \frac{\eta}{4}= \frac{\eta}{2} \;. \label{eq:a8}
    \end{align}
    We now justify each step in the above derivations. \eqref{eq:a2} uses the Gaussian concentration inequality $\Pr[ z^t > \beta] \leq e^{-\beta^{2/t}/2}$ for $\beta = i_{\max} s = 10\log^k(1/\eta)$.
    \eqref{eq:a3} holds because of \eqref{eq:ratio}.
    \eqref{eq:a5} follows because $i s\leq  10\log^k(1/\eta)$.
    \eqref{eq:a7} uses the Gaussian moment bound.
    \eqref{eq:a8} holds because $(t-1)!!\leq t^t\leq k^k$ and the choice $s:= 0.01\eta/(k^k \log^k(1/\eta))$.
    
    The other direction, i.e., $\int_{0}^{+\infty} g_{\eta}(z) z^t \phi(z) \d z  \geq - \eta/2$, can be shown with a similar argument:
    \begin{align*}
       \int_{0}^{+\infty} g_{\eta}(z) z^t \phi(z) \d z &\geq \sum_{i=0}^{i_{\max}-1}\left\{ -\delta\int_{(i+\delta)s}^{(i+1)s} z^t \phi(z) \d z + (1-\delta) \int_{(i+1)s}^{(i+1+\delta)s} z^t \phi(z) \d z \right\} \\
       &\geq \sum_{i=0}^{i_{\max}-1}\left\{ -\delta\left( \frac{1-\delta}{\delta}\right)(1+ 2 i s^2) +1-\delta\right\}\int_{(i+1)s}^{(i+1+\delta)s} z^t \phi(z) \d z \\
       &\geq - \sum_{i=0}^{i_{\max}-1} 2(1-\delta) i s^2 \int_{(i+1)s}^{(i+1+\delta)s} z^t \phi(z) \d z 
       \geq -20 s \log^k(1/\eta) (t-1)!! \geq -\frac{\eta}{4}\;,
    \end{align*}
    where instead of \eqref{eq:ratio} we used the bound $\int_{(i+\delta)s}^{(i+1)s} z^t \phi(z) \d z\leq \frac{1-\delta}{\delta}(1+2i s^2) \int_{(i+1)s}^{(i+1+\delta)s} z^t \phi(z) \d z$, which can be shown in a similar manner.

    We finally calculate the $\Pr[z \in U]$ (which is the same as $\Pr[g_{\eta}(z)=1-\delta]$), as follows 
    \begin{align*}
        \Pr[z \in U] &\leq 2 \sum_{i=0}^{i_{\max}-1} \int_{is}^{(i+\delta)s} \phi(z) \d z = 2\sum_{i=0}^{i_{\max}-1} \delta(1+ 2 i s^2) \int_{(i+1)s}^{(i+2)s}   \phi(z) \d z\\
        &=2 \delta \sum_{i=0}^{i_{\max}-1} \int_{(i+1)s}^{(i+2)s}   \phi(z) \d z
        + 4 \delta \sum_{i=0}^{i_{\max}-1} i s^2 \int_{(i+1)s}^{(i+2)s}   \phi(z) \d z\\
        &\leq 2 \delta \sum_{i=0}^{i_{\max}-1}\int_{(i+1)s}^{(i+2)s}   \phi(z) \d z + 40 s \log^k(1/\eta)\sum_{i=0}^{i_{\max}-1}\int_{(i+1)s}^{(i+2)s}   \phi(z) \d z\\
        &\leq  2 \delta \frac{1}{2} +  20 s \log^k(1/\eta) \frac{1}{2} \leq \delta + \eta \;,
    \end{align*}
    where the first line uses a ratio bound similar to \eqref{eq:ratio}, the third line uses that $i s \leq 10 \log^k(1/\eta)$, and the last line uses that $\sum_i \int_{(i+1)s}^{(i+2)s}   \phi(z) \d z \leq 1/2$ and that $s:= 0.01\eta/(k^k \log^k(1/\eta))$.

    Similarly, it can be shown that $\Pr[z \in U] \geq \delta - \eta$, 
    which completes the proof. \qed

\section{Proof of \Cref{prop:matching2}}
\label{sec5}

The high-level approach for proving \Cref{prop:matching2} is to first show 
a relaxed version of the  statement, where the ``hard set $T$'' is replaced 
by a ``soft set $f$'' which is a function $f: \R \to [0,1]$. 
That is, define the distribution
\begin{align}
    P_f(x) =  \phi(x-\eps) (1 - \delta f(x) )Z^{-1} \;\; \text{for} \;\; Z:=\int_{-\infty}^{+\infty} \phi(x-\eps) (1 - \delta f(x)  )\d x \;.\label{eq:distr}
\end{align}
We seek to find an $f: \R \to [0,1]$ satisfying the following two constraints: 
\begin{enumerate}
    \item (Moment matching) $\E_{x \sim P_f}[x^t] = \E_{x \sim \cN(0,1)}[x^t]$, and \label{it:it1}
    \item ($f$ has small mass) $Z = 1 - \delta \eps^{0.3}$ \label{it:it2}
\end{enumerate}

\noindent Note that this is indeed a relaxed version of the statement of \Cref{prop:matching2} which results by replacing the $\1(x \in T)$ by $f(x)$: the first constraint above is the relaxed version of \eqref{eq:marginal-y} and the second constraint is equivalent to $\E_{x \sim \cN(\eps,1)}[f(x)] \leq \eps^{0.3}$, which is the relaxed version of the constraint 
$\Pr_{x \sim \cN(\eps,1)}[x \in T] \leq \eps^{0.3}$ appearing in  \Cref{prop:matching2}.
Once we find such an $f$, we can convert it to a ``hard set'' $T$ which is a union of intervals by using a randomized rounding technique, similar to \cite{DKZ20}. Finally, that technique does not ensure any guarantees on the number of intervals produced, but using \Cref{prop:main_stuct} as in the previous section, we can bring this number down to $k$.

We will prove \Cref{it:it1} and \Cref{it:it2} that were listed before in two steps: We will find an $f$ consisting of two parts $f(x) = f_1(x) + f_2(x)$ with $f_1(x) \in [\eps,1/2]$ and $f_2(x) \in [-\eps,\eps]$ (so that overall $f(x) \in [0,1])$. For the first part (cf. \Cref{cl:explicit_calc}), the idea is to start by $f_1$ being the function that would make the distribution $P_{f_1}$ (cf. notation of \eqref{eq:distr}) exactly the same as $\phi(x)$ (the pdf of $\cN(0,1)$), and then clip $f_1(x)$ so that it only takes values in $[\eps,1/2]$. The important observation is that the clipping only happens for $x$ with large $|x|$. Thus, already $P_{f_1}$ is equal to $\phi(x)$ on big part of the real line. The remaining part contributes negligible amount to the moments, thus we can correct the moments by adding a correction term $f_2(x)$ to $f_1(x)$. We find $f_2$ by finding an appropriate polynomial using a technique from \cite{DKS17-sq}.

We now implement the two steps of the proof. For the first one, (regarding $f_1$), we have the following.

\begin{claim} \label{cl:explicit_calc}
    Fix $\delta = \sqrt{\eps}$. There exists an $f_1 : \R \to [\eps,1/2]$ such that $\int_{-\infty}^{+\infty} \phi(x-\eps)(1-\delta f_1(x)) \d x = 1-\delta \eps^{0.3}$  and the distribution with pdf
    \begin{align*}
        P_{f_1}(x) = \frac{\phi(x-\eps)(1-\delta f_1(x))}{1-\delta \eps^{0.3}}
    \end{align*}
    satisfies $q(x)=\phi(x)$ for all $x$ with $|x| \leq 1/\eps^{2/11}$.
\end{claim}
\begin{proof}
Define $\xi : = \delta \eps^{0.3} = \eps ^{0.8}$ (recalling that $\delta = \sqrt{\eps})$.
    For notational convenience, we will consider the following equivalent statement of our claim: there exists an $h: \R \to [1-\delta\eps,1 - \delta/2]$ such that $\int_{-\infty}^{+\infty} \phi(x-\eps) h(x) \d x = 1-\xi=1-\delta \eps^{0.3}$ and  $\frac{\phi(x-\eps)h(x)}{1-\xi}=\phi(x)$ for all $x$ with $|x| \leq 1/\eps^{2/11}$; the original statement would follow by this after letting $f_1(x) = (1-h(x))/\delta$).

    To show our claim, let us first consider the function $\tilde{h}$, which we define so that
    \begin{align*}
        \frac{\phi(x-\eps)\tilde h(x)}{1-\xi} = \phi(x)  \quad \text{for all $x \in \R$}\;.
    \end{align*}
    That is, we define $\tilde h(x) := \exp(\eps^2/2-\eps x)(1-\xi)$. Then, we define $h$ to the version of $\tilde h$ which is clipped in the interval $[1-\delta\eps,1-\delta/2]$, i.e., 
    \begin{align*}
        h(x) := 
        \begin{cases}
            1-\delta \eps , &\text{if $\tilde{h}(x) > 1-\delta\eps$} \\
            \tilde{h}(x)&\text{if $1-\delta/2 \leq \tilde{h}(x) \leq  1-\delta\eps$} \\
            1-\delta/2   , &\text{if $\tilde{h}(x) < 1-\delta/2$}  \;.
        \end{cases}
    \end{align*}
    Finally, it remains to verify that the clipping happens only for $x$ with  $|x| > 1/\eps^{2/11}$. First, note that $\tilde{h}$ is  a decreasing function. By plugging $x = - 1/\eps^{2/11}$ we can see that $h(- 1/\eps^{2/11}) = 1 - \Theta(\eps^{4/5})$ (we can see that by using a polynomial approximation for the $e^x$ function), which is less than the clipping threshold of $1-\delta \eps = 1- \eps^{1.5}$. Thus, by monotonicity of $h$, $\sup \{x \in \R: h(x) > 1-\delta \eps \} <   - 1/\eps^{2/11}$.
    Similarly, we can check the other boundary.
\end{proof}

We now move to the second part of the argument, which aims to find an $f_2 : \R \to [-\eps,\eps]$ such that when $f=f_1+f_2$, the moments of $P_f$ get corrected and equal to those of $\cN(0,1)$.
Fix $C = \sqrt{\eps}$ and $\xi = \delta \eps^{0.3} = 1-\eps^{0.8}$. We will search for an $f_2$ of the particular form below
\begin{align}
    f_2(x) = \frac{1-\xi}{\delta} \frac{p(x)}{\phi(x-\eps)} \1(|x| \leq C) \;,  \label{eq:definitionf2}
\end{align}
for some appropriate polynomial with $\int_{-C}^C p(x) \d x = 0$ and small $|p(x)|$ for all $x \in [-C,C]$.
We now show how to find that polynomial and ensure the above properties. Our 
moment-matching constraint is the following 
(note that the normalization of the distribution is still $1-\xi$, 
because of the property $\int_{-C}^C p(x) \d x = 0$): 
\begin{align*}
    \int_{-\infty}^{+\infty} x^t \frac{\phi(x-\eps)(1-\delta f_1(x) - \delta f_2(x))}{1-\xi} \d x  = \int_{-\infty}^{+\infty} x^t \phi(x) \d x \;.
\end{align*}
By letting $P_{f_1}(x) = \frac{\phi(x-\eps)(1-\delta f_1(x))}{1-\xi}$ as in \Cref{cl:explicit_calc}, the above is equivalent to
\begin{align}\label{eq:momentmatchingconstr}
    \int_{-C}^{+C} x^t p(x) \d x = \int_{-\infty}^{+\infty} x^t P_{f_1}(x) \d x - \int_{-\infty}^{+\infty} x^t \phi(x) \d x\;.
\end{align}
The rest of the proof mirrors that in~\cite{DKS17-sq}.
By Claim 5.8 in \cite{DKS17-sq}, there exists a unique polynomial $p$ satisfying \eqref{eq:momentmatchingconstr}, which has the form $p(x) = \sum_{i=0}^ka_i P_i(x/C)$, where $P_i$ is the $i$-th Legendre polynomial and $a_i = \frac{2i+1}{2C} \int_{-C}^C P_i(x/C) p(x) \d x$. 
We want to show that $|a_i| = O( i\eps^{5})$. 
First we note why this would be enough. This is because, 
by properties of the Legendre polynomials (see \Cref{app:sq} for basic properties that we will use), 
it would imply that $|p(x)| = O(\sum_{i=1}^k |a_i|) = O(k^{2}\eps^{5})$ for all $x \in [-C,C]$. We would then be done, because after combining with \eqref{eq:definitionf2}, we would obtain that for all $x \in [-C,C]$ it holds
\begin{align*}
    |f_2(x)| \leq \frac{ (1-\xi)|p(x)|}{\delta\phi(x-\eps)}
    \leq \frac{O(\eps^{5} k^{2})}{\delta\phi(x-\eps)}
    <\eps  \;,
\end{align*}
where we used $\delta= \sqrt{\eps} $, $k\leq c/\eps^{0.15}$, and that  $\phi(x-\eps) \geq 1/3$ for all $x \in [-1,1]$.
We conclude by showing that $|a_i| = O( i\eps^{5})$.
First, by orthogonality of the $P_i$'s and \eqref{eq:momentmatchingconstr},
\begin{align*}
      \left| \int_{-C}^C \hspace{-0.4em}P_i(x/C) p(x) \d x \right|  &= \left| \int_{-\infty}^{+\infty} \hspace{-0.4em}P_i(x/C) (\phi(x) - P_{f_1}(x)) \d x \right| = \left| \int_{|x| > 1/\eps^{2/11}}\hspace{-2.6em} P_i(x/C) (\phi(x) - P_{f_1}(x)) \d x \right| \\
      &\leq \left| \int_{|x| > 1/\eps^{2/11}}\hspace{-2.5em} P_i(x/C)  \phi(x)  \d x \right| + \left| \int_{|x| > 1/\eps^{2/11}}\hspace{-2.5em} P_i(x/C) P_{f_1}(x)  \d x \right| \;,
\end{align*}
where the second step used that $P_{f_1}(x)=\phi(x)$ for all $x$ with $|x|\leq 1/\eps^{2/11}$ by \Cref{cl:explicit_calc}.
We will show the bound for the first term (the other one is similar).

\begin{claim}\label{cl:moment_negligible}
Fix $C=\sqrt{\eps}$, and let $P_i$ denote the $i$-Legendre polynomial and $p$ be the solution to \eqref{eq:momentmatchingconstr}. Then, $\left| \int_{|x| > 1/\eps^{2/11}} P_i(x/C)  \phi(x)  \d x \right| = O(\eps^5)$.
\end{claim}
\begin{proof}
We will use the known property that the $j$-th Legendre polynomial can be written as $P_j(x) = \frac{1}{2^i} \sum_{j=0}^{\lfloor i/2 \rfloor} \binom{i}{j}\binom{2i-2j}{i}x^{i-2j}$. We will also use that 
there is negligible mass at the tails $|x|>1/\eps^{2/11}$. We have that
    \begin{align*}
        \left|   \int_{-C}^C P_i(x/C) p(x) \d x \right| &= O((k/C)^{3k})  \int_{|x| > 1/\eps^{2/11}} \hspace{-1.5em}|x|^k \phi(x) \d x \\
        &\leq O((k/C)^{3k}) \int_{|x| > 1/\eps^{2/11}} \hspace{-1.5em}|x|^k e^{-x^2/2} \d x\\
        &\leq O((k/C)^{3k}) \eps^{10 k} = O((k/\sqrt{\eps})^{3k}) \eps^{10 k} =O( \eps^{5}) \;,
    \end{align*}
    where the first step bounds the binomial coefficients by $k^k$
    and in the last line uses that  for any $|x| > 100 k \log(1/\eps) = \eps^{-3/20}\log(1/\eps)$ (recall that $k=\Theta(\eps^{3/20})$) it holds $ |x|^k e^{-x^2/2} < \eps^{10k}/x^2$, in order to bound $\int_{|x| > \eps^{-2/11}} |x|^k e^{-x^2/2} \d x\leq \int_{|x| > 1} |x|^k e^{-x^2/2} \d x \leq \eps^{10k} \int_{|x| > 1}x^{-2}\d x \leq \eps^{10k}$.
    
\end{proof}

This completes the proof of  \Cref{it:it1,it:it2}. We next use a randomized rounding technique similar to \cite{DKZ20}, in order to convert this continuous $f$ to a piecewise constant $\tilde f: \R \to \{0,1\}$, i.e., a hard set. We show the following in \Cref{app:sec5}:

\begin{restatable}{claim}{DISCRETIZATION}\label{lem:function-second}
For any $\eta>0$ there exists a $((k\log(1/\eta))^{\poly(k)}/\eta^2)$-piecewise constant function
$\tilde f:\R\to\{0,1\}$ such that 
    $\Pr_{x \sim \cN(\eps,1)}[\tilde f(x)] \leq 2 \delta$ and for all $t=0,\ldots,k$ it holds $|\E_{x \sim \cN(\eps,1)}[x^t (1-\delta \tilde f(x))] Z^{-1} - \E_{x \sim \cN(0,1)}[x^t]| \leq \eta$, where $Z= \E_{x \sim \cN(\eps,1)}[1-\delta \tilde f(x)]$.
\end{restatable}

The idea for \Cref{lem:function-second} is to split $\R$ into $[is,(i+1)s]$, for $i\in \Z$ and a sufficiently small size $s$, and
to let $\tilde{f}(x)$ be constant in the interval $x \in [is,(i+1)s)$, taking the following values:
\begin{align}
    \tilde{f}(x) = 
    \begin{cases}
        1 , &\text{with probability $p_i:= \int_{is}^{(i+1)s} \phi(x-\eps)f(x)\d x/\int_{is}^{(i+1)s} \phi(x-\eps)\d x$}  \\
        0, &\text{with probability $1-p_i$} 
    \end{cases}
\end{align}
We want to show that $\E_{x \sim \cN(\eps,1)}[x^t(1-\delta \tilde f(x))] \approx \E_{x \sim \cN(\eps,1)}[x^t(1-\delta f(x))]$ (which we have already shown that is equal to $\E_{x \sim \cN(0,1)}[x^t]$).
Let $I_i := \int_{is}^{(i+1)s}x^t\phi(x-\eps) \delta (f(x) - \tilde{f}(x))\d x$ be the contribution due to the $i$-th interval. Then, using the Taylor approximation $x^t = (is)^t + (x-is) t \xi^{t-1}$ for some $\xi$ between $is$ and $x$, the expected (with respect to $\tilde f$'s randomness) value of  $\sum_i I_i$ is
\begin{align*}
    \E\bigg[\sum_i I_i\bigg] = \sum_i (is)^t \int_{is}^{(i+1)s} \hspace{-18pt}\phi(x-\eps)\delta(f(x)-p_i)\d x + t \xi^{t-1} \int_{is}^{(i+1)s} \hspace{-18pt}(x-is)\phi(x-\eps)\delta(f(x)-p_i)\d x .
\end{align*}
The first term above is zero by definition of the $p_i$'s. We can show that the second term is at most $\eta$ by choosing appropriately small interval size $s$.

The proof of \Cref{prop:matching2} is completed by reducing the number of pieces to $k$ using \Cref{lem:function-second} and \Cref{prop:main_stuct} as we did in \Cref{prop:matching1}.

\subsection{Proof of \Cref{lem:function-second}}\label{app:sec5}

We fix the following parameters throughout the proof (where $C$ denotes a sufficiently large absolute constant):
\begin{itemize}
    \item $i_{\max} = (C\log(1/\eta))^{k/2}/s$
    \item $s =  \eta^2/(k^{3k} C^{2k^2}\log^{k^2}(1/\eta))$    
\end{itemize}
We partition the real line in pieces $[is,(i+1)s)$ for $i \in \Z$. We define $\tilde{f}$ to be the following random piecewise-constant function: For each $i \in \{ -i_{\max},\ldots,i_{\max}\}$ we let $\tilde{f}(x)$ be constant in the interval $x \in [is,(i+1)s)$, taking the following value:
\begin{align}\label{eq:def_f}
    \tilde{f}(x) = 
    \begin{cases}
        1 , &\text{with probability $p_i:= \int_{is}^{(i+1)s} \phi(x-\eps)f(x)\d x/\int_{is}^{(i+1)s} \phi(x-\eps)\d x$}  \\
        0, &\text{with probability $1-p_i$} 
    \end{cases}
\end{align}
 and  we define $\tilde{f}(x)=0$ with probability 1 in the entire $(-\infty,-i_{\max}s) \cup [i_{\max}s,+\infty)$.

Our goal is to show that for all $t=0,\ldots,k$, we have  $|\int_{\R} x^t P_f(x) \d x - \int_{\R} x^t P_{\tilde{f}}(x) \d x | \ll  \eta$, where we are using the notation from \eqref{eq:distr}. We will do this in two steps: we will first show that $\int_{\R}x^t \phi(x-\eps)(1-\delta \tilde{f}(x)) \d x$ is approximately (up to an additive term of $\eta$) equal to $\int_{\R}x^t \phi(x-\eps)(1-\delta f(x)) \d x$ and then we will do the same for the normalizing factor  $\int_{\R} \phi(x-\eps)(1-\delta \tilde{f}(x)) \d x$.

We start with the first part, which we will do by probabilistic argument. 
First,
\begin{align}
   \int\limits_{-\infty}^{\infty} &x^t \phi(x-\eps)(1-\delta \tilde{f}(x)) \d x - \int\limits_{-\infty}^{\infty}  x^t \phi(x-\eps)(1-\delta f(x)) \d x \notag\\ 
   &=  
   \int\limits_{(i_{\max}+1)s}^{\infty} \hspace{-5pt}x^t\phi(x-\eps)\d x
   +\hspace{-8pt}\int\limits_{-\infty}^{-i_{\max}s} x^t\phi(x-\eps)\d x
   +  \hspace{-10pt}\sum_{i=-i_{\max}}^{i_{\max}} \int\limits_{is}^{(i+1)s}\hspace{-5pt}x^t\phi(x-\eps) \delta (f(x) - \tilde{f}(x))\d x \;. \label{eq:error_unormalized}
\end{align}
We note that the first two terms are negligible, i.e., less than a small multiple of $\eta$. This is because of the fact $\Pr_{z \sim \cN(0,1}[|z|^t > \beta] \leq e^{-\beta^2/2}$  for all $\beta\geq 1$, applied with $\beta = i_{\max} s = (C\log(1/\eta))^{k/2}$. 

For the remaining sum, let us use the notation $I_i := \int_{is}^{(i+1)s}x^t\phi(x-\eps) \delta (f(x) - \tilde{f}(x))\d x$. These are random integrals, where the randomness comes from how $\tilde{f}(x)$ is defined in $[is,(i+1)s)$. The goal is to show that with non-zero probability $|\sum_{i=-i_{\max}}^{i_{\max}} I_i | \ll \eta$. Then, by probabilistic argument we would know that such an $\tilde{f}$ exists. 

We start with the expectation of these $I_i$'s, where we will employ Taylor's theorem for $x^t$, i.e., $x^t = (is)^t + (x-is) t \xi^{t-1}$ for some $\xi=\xi(x)$ between $is$ and $x$. We have that:
\begin{align*}
   &\E\left[\sum_{i=-i_{\max}}^{i_{\max}} I_i \right] =  \sum_{i=-i_{\max}}^{i_{\max}} \int\limits_{is}^{(i+1)s}x^t\phi(x-\eps) \delta (f(x) - p_i)\d x \\
   &= \sum_{i=-i_{\max}}^{i_{\max}}(is)^t \int\limits_{is}^{(i+1)s}\phi(x-\eps) \delta (f(x) - p_i)\d x  + t \xi^{t-1}  \int\limits_{is}^{(i+1)s}(x-is)\phi(x-\eps) \delta (f(x) - p_i)\d x \;.
\end{align*}
The first term above is zero because of the definition of $p_i$ from \eqref{eq:def_f}. For the second term, we have the following bounds:
\begin{align*}
    \bigg| \sum_{i=-i_{\max}}^{i_{\max}} t \xi^{t-1} \int\limits_{is}^{(i+1)s}(x-is)\phi(x-\eps) & (f(x) - p_i)\d x  \bigg|
    \\&\leq t (i_{\max} s)^{t-1} \sum_{i=-i_{\max}}^{i_{\max}} \int\limits_{is}^{(i+1)s}   |x-is| \phi(x-\eps)\d x \\
    &\leq t (i_{\max} s)^{t-1} \sum_{i=-i_{\max}}^{i_{\max}}    s \int\limits_{is}^{(i+1)s}  \phi(x-\eps)\d x\\
    &\leq    s t (i_{\max} s)^{t-1} \leq  s t (C^{k/2}\log^{ k/2}(1/\eta))^{t-1} \ll \eta \;,
\end{align*}
where the first line uses that $\delta \leq 1$,  $\xi \leq i_{\max} s$, $f(x) \in [0,1]$, and $p_i \in [0,1]$,
the second line uses that the integral is over an interval of length $s$,
the third line first uses that $\int_{is}^{(i+1)s}\phi(x-\eps)\d x \leq \int_{-\infty}^{+\infty}\phi(x-\eps)\d x = 1$ and then 
uses that by our choice of parameters: first $i_{\max} s =(C\log(1/\eta))^{k/2}$, and finally $s =  \eta^2/(k^{3k} C^{2k^2}\log^{k^2}(1/\eta))$.
This completes the proof that $|\E[\sum_{i=-i_{\max}}^{i_{\max}} I_i]| \ll \eta$.

We now show the non-trivial probability claim. By the Chernoff-Hoeffding bound, with probability at least $1-\tau$, it holds $|\sum_{i=-i_{\max}}^{i_{\max}} I_i - \E[\sum_{i=-i_{\max}}^{i_{\max}} I_i] | \lesssim \Delta \sqrt{i_{\max} \log(1/\tau)}$ where $\Delta$ is any value such that $|I_i| \leq \Delta$ with probability one. In our case, we have that $|I_i| = | \int_{is}^{(i+1)s} x^t \phi(x-\eps) \d x | \leq s \cdot \sup_{x \in \R} x^t e^{-x^2} \leq s t^t \leq s k^k$. We also  use $\tau=0.1/k$ because we want the conclusion to hold simultaneously over all $t=0,\ldots,k$. Using these parameters,  and our definitions for $s$ and $i_{\max}$,  the application of Chernoff-Hoeffding bound yields that $|\sum_{i=-i_{\max}}^{i_{\max}} I_i - \E[\sum_{i=-i_{\max}}^{i_{\max}} I_i] | \leq  k^k s\sqrt{i_{\max}\log k} \leq k^k \sqrt{s} \sqrt{i_{\max} s} \sqrt{\log k}\leq \sqrt{s} (C \log(1/\eta))^{k/4} k^k \sqrt{\log k} \leq  \eta$.

We now move to the second (and easier) part of the proof regarding the normalizing factor. We want to show that $|\int_{\R} \phi(x-\eps)(1-\delta \tilde{f}(x)) \d x- \int_{\R} \phi(x-\eps)(1-\delta f(x)) \d x|\ll \eta$. As before, the parts of the integral in $(-\infty,-i_{\max}s) \cup [(i_{\max}+1)s,+\infty)$ do not mater (the error term $r$  has $|r| \ll \eta$ below) :
\begin{align*}
     \int\limits_{-\infty}^{+\infty}\phi(x-\eps)(1-\delta \tilde{f}(x)) \d x&- \int\limits_{-\infty}^{+\infty} \phi(x-\eps)(1-\delta f(x)) \d x  \\&\leq r+ \sum_{i=-i_{\max}}^{i_{\max}}\int\limits_{is}^{(i+1)s} \phi(x-\eps) \delta (f(x) - \tilde{f}(x)) \d x  \;.
\end{align*}
Re-define  $I_i := \int\limits_{is}^{(i+1)s} \phi(x-\eps) \delta (f(x) - \tilde{f}(x)) \d x$. By definition of $\tilde{f}$,  $\E[I_i]=0$ for all the pieces $i =-i_{\max},\ldots,  i_{\max}$. An application of Chernoff-Hoeffding bounds similar to the previous one also shows that $|\sum_{i=-i_{\max}}^{i_{\max}}I_i| \ll \eta$ with large constant probability.

The proof is now concluded by noting that
\begin{align}
    \int\limits_{-\infty}^{+\infty} x^t P_{\tilde{f}}(x)\d x &= \frac{\int\limits_{-\infty}^{+\infty} x^t \phi(x-\eps)(1-\delta \tilde{f}(x))\d x}{\int\limits_{-\infty}^{+\infty} \phi(x-\eps)(1-\delta \tilde{f}(x))\d x  }
    =\frac{\int\limits_{-\infty}^{+\infty} x^t \phi(x-\eps)(1-\delta \tilde{f}(x))\d x  \pm \eta/100}{\int\limits_{-\infty}^{+\infty} \phi(x-\eps)(1-\delta \tilde{f}(x))\d x  \pm \eta/100} \notag\\
    &= \frac{\int\limits_{-\infty}^{+\infty} x^t \phi(x-\eps)(1-\delta \tilde{f}(x))\d x \pm \eta/100}{(1\pm \eta/100)\int\limits_{-\infty}^{+\infty} \phi(x-\eps)(1-\delta \tilde{f}(x))\d x}
    =(1\pm \eta/2) \int\limits_{-\infty}^{+\infty} x^t P_{{f}}(x)\d x \pm \frac{\eta}{2} \;, \label{eq:final}
\end{align}
where the second line used that the normalizing factor is $\int\limits_{-\infty}^{+\infty} \phi(x-\eps)(1-\delta \tilde{f}(x))\d x = \Omega(1)$. Finally, if we used $\eta/k^k$ in place of $\eta$ everywhere from the beginning of this proof, we could make the RHS of \eqref{eq:final} at most $\int\limits_{-\infty}^{+\infty} x^t P_{{f}}(x)\d x \pm \eta$ (this is because $\int\limits_{-\infty}^{+\infty} x^t P_{{f}}(x)\d x $ is the same as the Gaussian moments).    \qed

\clearpage

\bibliographystyle{alpha}
\bibliography{allrefs}

\newpage

\appendix

\section*{Appendix}

\section{Additional Preliminaries}\label{app:sq}

\paragraph{Additional Notation}
We use $a\lesssim b$ to denote that there exists an absolute universal constant $C>0$ (independent of the variables or parameters on which $a$ and $b$ depend) such that $a \leq C  b$. We write $a \ll b$ to denote that $\alpha \leq c b$ for a sufficiently small  absolute constant $c>0$.

\paragraph{Legendre Polynomials} In this work, we make use of the Legendre Polynomials which are orthogonal polynomials over $[-1,1]$. Some of their properties are:

\begin{fact}[\cite{Sze67}] The Legendre polynomials $P_k$ for $k\in\Z$, satisfy the following properties:
\begin{enumerate}
    \item $P_k$ is a $k$-degree polynomial and $P_0(x)=1$ and $P_1(x)=x$.
    \item $\int_{-1}^1 P_i(x)P_j(x) \d x=2/(2i+1)\1\{i=j\}$, for all $i,j\in \Z$.
    \item $|P_k(x)|\leq 1$ for all $|x|\leq 1$.
    \item $P_k(x)=(-1)^k P_k(-x)$.
    \item $P_k(x)=2^{-k}\sum_{i=1}^{\lceil k/2 \rceil}\binom{k}{i}\binom{2k-2i}{k}x^{k-2i}$.
\end{enumerate}
    
\end{fact}

\paragraph{Additional Background on the SQ Model}
We now record additional definitions and facts from \cite{FGR+13} that are relevant to the SQ model.

\begin{definition}[Pairwise Correlation] \label{def:pc}
The pairwise correlation of two distributions with probability density functions
$D_1, D_2 : \R^d \to \R_+$ with respect to a distribution with
density $D: \R^d \to \R_+$, where the support of $D$ contains
the supports of $D_1$ and $D_2$, is defined as
$\chi_{D}(D_1, D_2) = \int_{\R^d} D_1(\bx) D_2(\x)/D(\bx)\, \d\bx - 1$.
\end{definition}

\begin{definition}[$\chi^2$-divergence]
The $\chi^2$-divergence between
$D_1, D_2 : \R^d \to \R_+$  is defined as
$\chi^2(D_1, D_2) = \int_{\R^d} D_1^2(\bx)/D_2(\bx)\, \d\bx - 1$.
\end{definition}

\begin{definition} \label{def:uncor}
We say that a set of $s$ distributions $\mathcal{D} = \{D_1, \ldots , D_s \}$
 is $(\gamma, \beta)$-correlated relative to a distribution $D$
if $|\chi_D(D_i, D_j)| \leq \gamma$ for all $i \neq j$,
and $|\chi_D(D_i, D_j)| \leq \beta$ for $i=j$.
\end{definition}

\begin{definition}[Decision Problem over Distributions] \label{def:decision}
Let $D$ be a fixed distribution and $\D$ be a distribution family.
We denote by $\mathcal{B}(\D, D)$ the decision (or hypothesis testing) problem
in which the input distribution $D'$ is promised to satisfy either
(a) $D' = D$ or (b) $D' \in \D$, and the goal
is to distinguish between the two cases.
\end{definition}

\begin{definition}[Statistical Query Dimension] \label{def:sq-dim}
Let $\beta, \gamma > 0$. Consider a decision problem $\mathcal{B}(\D, D)$,
where $D$ is a fixed distribution and $\D$ is a family of distributions. Define $s$ to be the maximum integer such that there exists a finite set of distributions
$\mathcal{D}_D \subseteq \D$ such that
$\mathcal{D}_D$ is $(\gamma, \beta)$-correlated relative to $D$
and $|\mathcal{D}_D| \geq s.$ The {\em Statistical Query dimension}
with pairwise correlations $(\gamma, \beta)$ of $\mathcal{B}$ is defined as $s$ and denoted as $\mathrm{SD}(\mathcal{B},\gamma,\beta)$.
\end{definition}

\begin{lemma}[Corollary 3.12 in \cite{FGR+13}] \label{lem:sq-from-pairwise}
Let $\mathcal{B}(\D, D)$ be a decision problem. For $\gamma, \beta >0$,
let $s= \mathrm{SD}(\mathcal{B}, \gamma, \beta)$.
For any $\gamma' > 0,$ any SQ algorithm for $\mathcal{B}$ requires queries of tolerance at most $\sqrt{\gamma + \gamma'}$ or makes at least
$s  \gamma' /(\beta - \gamma)$ queries.
\end{lemma}

We need the following result from \cite{diakonikolas2021optimality} that upper bounds the correlation between two such distributions.

\begin{lemma}[Corollary 2.4 in \cite{diakonikolas2021optimality}] \label{lemma:correlation-bound}
Let  $A$ be a distribution over $\R^m$ such that the first $k$ moments of $A$
match the corresponding moments of $\normal(\vec 0,\vec I_m)$.
For matrices $\vec U,  \vec V \in \R^{m\times d}$ such that $\vec U \vec U^\top =  \vec V \vec V^\top = \vec I_m$,
define $P_{A, \vec U}$ and $P_{A, \vec V}$ to be distributions over $\R^d$ according to \Cref{def:hidden}. Then, the following holds: 
$|\chi_{\normal(\vec 0,\vec I_m)}(P_{A, \vec U},P_{A, \vec V})| \leq \|\vec U\vec V^\top\|_\op^{k+1} \chi^2(A,\normal(\vec 0,\vec I_m))$.
\end{lemma}

\subsection{Proof of \Cref{lem:sq-hardness}}
We restate and prove the following fact.
\begin{fact}
Let $d,k\in \Z$ and $m<d^{1/10}$. Let  $A$ be a distribution over $\R^m$ such that the first $k$ moments of $A$
match the corresponding moments of $\normal(\vec 0,\vec I_m)$. Define the set $\D$ of distributions containing distributions constructed as follows: for matrices $\vec U\in \R^{m\times d}$ such that $\vec U \vec U^\top =\vec I_m$,
define $P_{A, \vec U}$ to be distributions over $\R^d$ according to \Cref{def:hidden}. Then, any statistical query algorithm that solves the decision problem $\mathcal{B}(\D,\mathcal{N}(\vec 0,\vec I_d))$, requires either $2^{d^{\Omega(1)}}$ many queries, or performs at least one query with tolerance $d^{-\Omega(k)}\chi^2(A,\mathcal{N}(\vec 0,\vec I_m))$.
\end{fact}

\begin{proof}
        Recall the definition of \emph{decision problems} (\Cref{def:decision}). Let the  decision problem $\mathcal{B}(\D, D)$ where $D = \cN(\vec 0, \bI_d)$ and $\D$ is defined as the in the alternative hypothesis class above.     
     We now lower bound the SQ dimension (\Cref{def:sq-dim}) of $\mathcal{B}(\D, D)$. Let $S'$ be the set of matrices from the fact below.
     \begin{fact}[See, e.g., Lemma 17 in \cite{diakonikolas2021optimality}   ] \label{fact:setofmatrices}
Let $m,d \in \N$ with $m<d^{1/10}$.  There exists a set $S$ of $2^{d^{\Omega(1)}}$ matrices in $\R^{m \times d}$ such that every $\bU \in S$ satisfies $\bU \bU^\top = \bI_m$ and every pair $\bU,\bV \in S$ with $\bU \neq \bV$ satisfies $\|\bU \bV^\top \|_\fr \leq O(d^{-1/10})$.
\end{fact}
     Let $\mathcal{D}_D := \{ P_{A, \bV} \}_{\bV \in S}$ for the distribution $A$. 

     Using \Cref{fact:setofmatrices} and \Cref{lemma:correlation-bound}, we have that for any distinct $\vec V, \vec U \in S$ 
    \begin{align}
        |\chi_{\normal(\vec 0,\vec I_d)}(P_{A, \vec U},P_{A, \vec V})| &\leq \left\|\vec U\vec V^\top  \right\|_\op^{k+1} \chi^2(A,\normal(\vec 0,\vec I_2))\leq \Omega(d)^{-(k+1)/10} \chi^2(A,\normal(\vec 0,\vec I_d)) \;, \label{eq:sub-optimal}
    \end{align}
    where we used that $\| \vec A \|_\op \leq \|\vec A \|_\fr$ for any matrix $\vec A$.
    On the other hand, when $\vec V = \vec U$, we have that $ |\chi_{\normal(\vec 0,\vec I_d)}(P_{A, \vec U},P_{A, \vec V})| \leq \chi^2(A,\normal(\vec 0,\vec I_d)) $. 
    Thus, the family $\mathcal{D}_D$ is 
    $(\gamma,\beta)$-correlated with $\gamma = \Omega(d)^{-(k+1)/10} \chi^2(A,\normal(\vec 0,\vec I_d)) $ and $\beta = \chi^2(A,\normal(\vec 0,\vec I_d)) $ with respect to $D=\cN(\vec 0, \bI_2)$.
    This means that $\mathrm{SD}(\mathcal{B}(\D, D),\gamma,\beta)\geq \exp({d^{\Omega(1)}})$. 
     Therefore, by applying \Cref{lem:sq-from-pairwise} with $\gamma' := \gamma = \Omega(d)^{-(k+1)/10} \chi^2(A,\normal(\vec 0,\vec I_d))$, we obtain that any SQ algorithm for $\mathcal{Z}$ requires at least $\exp(d^{\Omega(1)}) d^{-O(k)} = $ calls to 
    \begin{align*}
        \mathrm{STAT}\left( \Omega(d)^{-\Omega(k)}  \chi^2(A,\normal(\vec 0,\vec I_d))\right) \;.
    \end{align*}
    
\end{proof}

\section{Omitted Proofs from \Cref{sec4}}\label{app:sec4}

\subsection{Proof of \Cref{prop:main_stuct}}\label{app:reduce}
We restate and prove the following:
\REDUCEINTERVALS*

\begin{proof}
Note that, we can always transform the function $g_\eta:\R\mapsto\{a,b\}$ to a $g_{\eta}':\R\mapsto \{\pm 1\}$ that satisfies similar properties. We define $g_{\eta}'(z)\eqdef (2g_\eta(z)-a-b)/(b-a)$ and let $\nu_t'=2\nu_t/(b-a)+(a+b)/(b-a)\E_{z \sim D}[z^t]$ and $\eta'=\eta(2/(b-a))$. Hence, we have that for any $\eta'>0$, there exists an at most $\ell$-piecewise constant function $g'_{\eta'}: \R \to \{\pm 1\}$
such that $|\E_{z \sim D}[g'_{\eta}(z)z^t]-\nu_t'|\leq\eta'$ for every non-negative integer $t<k$. 
By applying \Cref{lem:main_diff} and \Cref{lemma:compact}, we obtain that there exists an at most $(k+1)$-piecewise constant function $f': \R \to \{\pm 1\}$
such that $\E_{z \sim D}[f'(z)z^t]=\nu'_t$, for every non-negative integer $t<k$. By setting $f(z)=(f'(z)(b-a) +a+b)/2$, we complete the proof of \Cref{prop:main_stuct}.
\end{proof}
\begin{lemma}\label{lemma:compact}
Let $k$ be a positive integer. Let $D$ be a continuous distribution over $\R$ and let $\nu_0,\ldots,\nu_{k-1}\in\R$. If for any $\eta>0$ there exists an at most $(k+1)$-piecewise constant function $g_{\eta}: \R \to \{\pm 1\}$
such that $|\E_{z \sim D}[g_{\eta}(z)z^t]-\nu_t|\leq\eta$ for every non-negative integer $t<k$, then there exists an at most $(k+1)$-piecewise constant function $f: \R \to \{\pm 1\}$
such that $\E_{z \sim D}[f(z)z^t]=\nu_t$, for every non-negative integer $t<k$.
\end{lemma}
Lemma~\ref{lemma:compact} follows from the above using a compactness argument.

\begin{proof}
Let $p(z)$ be the pdf of $D$.
	For every $\eta>0$, we have that there exists a function $g_{\eta}$ such that $|\E_{z
		\sim D}[f_\eta(z)z^t]-\nu_t|\leq\eta$, for every non-negative integer $t<k$
	and the function $g_{\eta}$ is at most $(k+1)$-piecewise constant.
	Let $\vec M: \mathbb{\overline R}^{k} \mapsto \R^{k}$, where $M_i(\vec
	b)=\sum_{n=0}^{k}(-1)^{n+1}\int_{b_n}^{b_{n+1}} z^i p(z) \d z$ and
	$b_1\leq b_2\leq \ldots \leq b_{k}$, $b_0=-\infty$ and $b_{k+1}=\infty$. Here we assume
	without loss of generality that before the first breakpoint the function is
	negative because  we can always set the first breakpoint to be $-\infty$. It is
	clear that the function $\vec M$ is a continuous map and $\mathbb{\overline R}^{k+1}$ is a compact set,
	thus $\vec M\left(\mathbb{\overline R}^{k+1}\right)$ is a compact set.
	We also have that for every $\eta>0$ there is a point $\vec b\in \mathbb{\overline R}^{k+1}$
	such that $|\vec M(\vec b)\cdot\vec e_i -\nu_i|\leq \eta$ for all $i<k$. Thus, from compactness, we have that there
	exists a point $\vec b^*\in \mathbb{\overline R}^{k+1}$ such that $\vec M(\vec b^*)=\vec 0$.
	This completes the proof.
\end{proof}
The following lemma is similar with the main lemma of \cite{DKZ20}, we provide the proof for completeness as in our case the distributions are more general and we want specific values for their moments.
\begin{lemma}\label{lem:main_diff}
Let $m$ and $k$ be positive integers such that $m>k+1$ and $\eta>0$. Let $D$ be a continuous distribution over $\R$ and let $\nu_0,\ldots,\nu_{k-1}\in\R$.
If there exists an $m$-piecewise constant $f:\R \mapsto \{\pm 1\}$
such that $|\E_{z \sim D}[f(z)z^t]-\nu_t|<\eta$ for all non-negative integers $t<k$,
then there exists an at most $(m-1)$-piecewise constant $g :\R \mapsto \{\pm 1\}$
such that $|\E_{z \sim D}[g(z)z^t]-\nu_t|<\eta$ for all non-negative integers $t<k$.
\end{lemma}

\begin{proof}
Let $p(z)$ be the pdf of $D$.
Let $\{b_1,b_2,\ldots,b_{m-1}\}$ be the breakpoints of $f$, i.e., the points where the function $f$ changes value.
	Then let $F(z_1, z_2, \ldots, z_{m-1},z):\mathbb{\overline R}^m \mapsto \R$ be an $m$-piecewise constant function
	with breakpoints on $z_1, \ldots, z_{m-1}$, where $z_1<z_2< \ldots <z_{m-1}$
	and $F(b_1, b_2, \ldots, b_{m-1},z)=f(z)$.  For simplicity, let $\vec z=(z_1, \ldots, z_{m-1})$
	and  define $M_i(\vec z)= \E_{z \sim D}[F(\vec z, z)z^i]$ and let
	$\vec M(\vec z)=[M_0(\vec z), M_1(\vec z), \ldots M_{k-1}(\vec z)]^T$. It is
	clear from the definition that
	$M_i(\vec z)=\sum_{n=0}^{m-1}\int_{z_n}^{z_{n+1}} F(\vec z, z) z^i p(z) \d z =
	\sum_{n=0}^{m-1}a_n\int_{z_n}^{z_{n+1}} z^ip(z) \d z$,
	where $z_0= -\infty$ and $z_m=\infty$ and $a_n$ is the sign of $F(\vec z,z)$ in the interval $(z_n,z_{n+1})$.
	Note that $a_n=-a_{n+1}$ for every $0\leq n<m$.
	By taking the derivative of $M_i$ in $z_j$, for $0<j<m$, we get that
	\[\frac{\partial}{\partial z_j} M_i(\vec z)= 2a_{j-1} z_j^i p(z_j) \quad \text{and}\quad
	\frac{\partial}{\partial z_j} \vec M(\vec z)= 2a_{j-1} p(z_j) [1, z_j^1, \ldots ,z _j^{k-1}]^T\;.\]
	We now argue that for any $\vec z$ with distinct coordinates that there exists a vector $\vec u\in \R^{m-1}$ such that
	$\vec u=(\vec u_1,\ldots,\vec u_k,0,0,\ldots,0,1)$ and the directional derivative of $\vec M$ in the $\vec u$ direction
	is zero. To prove this, we construct a system of linear equations such that
	$\nabla_{\vec u} M_i(\vec z)=0$, for all $0\leq i<k$. Indeed, we have
$\sum_{j=1}^{k} \frac{\partial}{\partial z_j}  M_i(\vec z) \vec u_j
	= - \frac{\partial}{\partial z_{m-1}}  M_i(\vec z) $ or $\sum_{j=1}^{k} a_{j-1} z_j^i p(z_j)\vec u_j=- a_{m-2} z_{m-1}^i p(z_{m-1})$,
	 which is linear in the variables $\vec u_j$. Let $\hat{\vec u}$ be the vector with the first $k$ variables 
	 and let $\vec w$ be the vector of the right hand side of the system, i.e., $\vec w_i=- a_{m-2} z_{m-1}^i p(z_{m-1})$. Then
	 this system can be written in matrix form as $\vec V \vec D\hat{ \vec u}=\vec w$, where $\vec V$ is the Vandermonde matrix,
	 i.e., the matrix that is $\vec V_{i,j}=\alpha_i^{j-1}$, for some values $\alpha_i$ and $\vec D$ is a diagonal matrix.
	 In our case, $\vec V_{i,j}=z_i^{j-1}$ and $\vec D_{j,j}= 2 a_{j-1}p(z_j)$.
	 It is known that the Vandermonde matrix has full rank iff for all $i\neq j$ we have $\alpha_i\neq \alpha_j$,
	 which holds in our setting. Thus, the matrix $\vec V \vec D$ is nonsingular and there exists a solution to the equation.
Thus, there exists a vector $\vec u$ with our desired properties and, moreover,
	any vector in this direction is a solution of this system of linear equations.
	Note that the vector $\vec u$ depends on the value of $\vec z$,
	thus we consider $\vec u(\vec z)$ be the (continuous) function that returns a vector $\vec u$ given $\vec z$.
	
We define a differential equation for the function $\vec v:\mathbb{\overline R}\mapsto\mathbb{\overline R}^{m-1}$, as follows: $\vec v(0)= \vec b$, where $\vec b=(b_1, \ldots, b_{m-1})$, and
$ \vec v'(T)=\vec u(\vec v(T))$ for all $T \in \mathbb{\overline R}$.
If $\vec v$ is a solution to this differential equation, then we have:
\[\frac{\d}{\d T} \vec M(\vec v(T))=\frac{\d}{\d \vec v(T)} \vec M(\vec v(T)) \frac{\d}{\d T} \vec v(T)
=\frac{\d}{\d \vec v(T)} \vec M(\vec v(T)) \vec u(\vec v(T)) =\vec 0\;,
\]
where we used the chain rule and that the directional derivative in $\vec u(\vec v(T))$ direction is zero.
This means that the function $\vec M(\vec v(t))$ is constant, and for all $0\leq j<k$, we have $|M_j-\nu_j|< \eta$, because we have that $|\E_{z \sim D}[F(z_1,\ldots, z_{m-1},z)z^t]-\nu_t|<\eta$. Furthermore, since $\vec u(\vec v(T))$ is continuous in $\vec v(T)$, this differential equation will be well founded and have a solution up until the point where either two of the $z_i$ approach each other or one of the $z_i$ approaches plus or minus infinity (the solution cannot oscillate, since $\vec v_{m-1}'(T)=1$ for all $T$).
	
Running the differential equation until we reach such a limit, we find a limiting value $\vec v^\ast$ of $\vec v(T)$ so that either:
\begin{enumerate}[leftmargin=*]
\item There is an $i$ such that $\vec v_i^\ast=\vec v_{i+1}^\ast$, which
gives us a function that is at most $(m-2)$-piecewise constant, i.e., taking $F(\vec v^\ast,z)$.
\item Either $\vec v_{m-1}^\ast = \infty$ or $\vec v_1^\ast = -\infty$, which gives us an at most
$(m-1)$-piecewise constant function, i.e., taking $F(\vec v^\ast,z)$.
Since when the $\vec v_{m-1}^\ast= \infty$, the last breakpoint becomes $\infty$, we have one less breakpoint, and if $\vec v_1^\ast =-\infty$ we lose the first breakpoint.
\end{enumerate}
Thus, in either case we have a function with at most $m-1$ breakpoints and the same moments.
This completes the proof.
\end{proof}

\section{Lower Bounds for Low-Degree Polynomial Tests}\label{sec:low_degree}

We describe the implications of SQ lower bounds to  low-degree polynomials for the problem below:

\begin{problem} \label{def:nongaussian_new}
	Let a distribution $A$ on $\R^m$. For a matrix $\bV \in \R^{m \times d}$, we let $P_{A,\bV}$ be the distribution as in \Cref{def:hidden}, 
i.e., the distribution that coincides with $A$ on the subspace spanned by the rows of $\bV$ and is standard Gaussian in the orthogonal subspace. Let $S$ be the set of nearly orthogonal vectors from \Cref{fact:setofmatrices}. Let $\cS = \{ P_{A,v} \}_{u \in S}$. We define the simple hypothesis testing problem where the null hypothesis is $\mathcal{N}(\vec 0,I_d)$ and the alternative hypothesis is $P_{A,\bV}$ for some $\bV$ uniformly selected from $S$.
\end{problem}

 We now describe the model in more detail. We will consider tests that are thresholded polynomials of low-degree, i.e., output $H_1$ if the value of the polynomial exceeds a threshold and $H_0$ otherwise. We need the following notation and definitions.
For a distribution $D$ over $\cX$, we use $D^{\otimes n}$ to denote the joint distribution of $n$ i.i.d.\ samples from $D$.
For two functions $f:\cX \to \R$, $g: \cX \to R$ and a distribution $D$, we use $\langle f, g\rangle_{D}$ to denote the inner product $\E_{X \sim D}[f(X)g(X)]$.
We use $ \|f\|_{D}$ to denote $\sqrt{\langle f, f \rangle_{D} }$.
We say that a polynomial $f(x_1,\dots,x_n):\R^{n \times d} \to \R$ has sample-wise degree $(r,\ell )$ if each monomial uses at most $\ell$ different samples from $x_1,\dots,x_n$ and uses degree at most $r$ for each of them.
Let $\cC_{r,\ell}$ be linear space of all polynomials of sample-wise degree $(r,\ell)$ with respect to the inner product defined above.
For a function $f:\R^{n \times d} \to \R$, we use $f^{\leq r, \ell}$ to be the orthogonal projection onto $\cC_{r,\ell}$   with respect to the inner product $\langle \cdot , \cdot \rangle_{D_0^{\otimes n}}$.  Finally, for the null distribution $D_0$ and a distribution $P$, define the likelihood ratio $\overline{P}^{\otimes n}(x) := {P^{\otimes n}(x)}/{D_0^{\otimes n}(x)}$.

\begin{definition}[$n$-sample $\tau$-distinguisher]
For the hypothesis testing problem between  $D_0$ (null distribution) and $D_1$ (alternate distribution) over $\cX$,
we say that a function $p : \cX^n \to \R$ is an $n$-sample $\tau$-distinguisher if $|\E_{X \sim  D_0^{\otimes n}}[p(X)] - \E_{X \sim D_1^{\otimes n}}[p(X)]| \geq \tau \sqrt{\Var_{X \sim D_0^{\otimes n}} [p(X)] }$. We call $\tau$ the \emph{advantage} of the polynomial $p$. 
\end{definition}
Note that if a function $p$ has advantage $\tau$, then the Chebyshev's inequality implies that one can furnish a test $p':\cX^n \to \{D_0,D_1\}$ by thresholding $p$ such that the probability of error under the null distribution is at most $O(1/\tau^2)$. 
We will think of the advantage $\tau$ as the proxy for the inverse of the probability of error (see  Theorem 4.3 in \cite{kunisky2022notes}  for a formalization of this intuition under certain assumptions) and we will show that the advantage of all polynomials up to a certain degree is $O(1)$. 
It can be shown that for hypothesis testing problems of the form of   \Cref{def:nongaussian_new}, 
 the best possible advantage among all polynomials in $\cC_{r,\ell}$ is captured by the low-degree likelihood ratio (see, e.g.,~\cite{BBHLS20,kunisky2022notes}):
\begin{align*}
    \left\| \E_{v \sim \cU(S)}\left[ \left( \overline{P}_{A,\bV}^{\otimes n}  \right)^{\leq r, \ell } \right]  - 1  \right\|_{D_0^{\otimes n}},
\end{align*}
where in our case $D_0 = \cN(\vec 0,\bI_d)$.

To show that the low-degree likelihood ratio is small, we use the result from \cite{BBHLS20} stating that a lower bound for the SQ dimension translates to an upper bound for the low-degree likelihood ratio. Therefore, given that we have already established in  previous section that $\mathrm{SD}(\cB(\{P_{A,\bV} \}_{\bV \in S},\cN(\vec 0,\bI_d)), \gamma,\beta)=2^{d^c}$ for $\gamma=\Omega(d)^{(t+1)/10}\chi^2(A,\cN(\vec 0,\bI_d))$ and $\beta= \chi^2(A,\cN(0,1))$, we one can obtain the corollary:

        \begin{theorem}\label{cor:low-deg-hardness-general-problem} 
        Let a sufficiently small positive constant $c$. Let the hypothesis testing problem of  \Cref{def:nongaussian_new}  the distribution $A$ matches the first $t$  moments with $\cN(\vec 0, \bI_m)$. For any $d \in \Z_+$ with $d = t^{\Omega(1/c)}$, any $n \leq \Omega(d)^{(t+1)/10}/\chi^2(A,\normal(\vec 0,\vec I_m))$ and any even integer $\ell < d^{c}$, we have that
        \begin{align*}
            \left\| \E_{v \sim \cU(S)}\left[ \left( \overline{P}_{A,\bV}^{\otimes n}  \right)^{\leq \infty, \ell } \right]  - 1  \right\|_{D_0^{\otimes n}} \leq 1\;.
        \end{align*}
    \end{theorem}

The interpretation of this result is that unless the number of samples used $n$ is greater than $\Omega(d)^{(t+1)/10}/\chi^2(A,\normal(\vec 0,\vec I_m))$, any polynomial of degree roughly up to $d^{c}$  fails to be a good test (note that any polynomial of degree $\ell$ has sample-wise degree at most $(\ell,\ell)$).

\end{document}